\PassOptionsToPackage{unicode}{hyperref}
\PassOptionsToPackage{hyphens}{url}
\PassOptionsToPackage{dvipsnames,svgnames,x11names}{xcolor}
\documentclass[
  12pt]{article}

\usepackage{amsmath,amssymb}
\usepackage{iftex}
\ifPDFTeX
  \usepackage[T1]{fontenc}
  \usepackage[utf8]{inputenc}
  \usepackage{textcomp} 
\else 
  \usepackage{unicode-math}
  \defaultfontfeatures{Scale=MatchLowercase}
  \defaultfontfeatures[\rmfamily]{Ligatures=TeX,Scale=1}
\fi
\usepackage{lmodern}
\ifPDFTeX\else  
\fi
\IfFileExists{upquote.sty}{\usepackage{upquote}}{}
\IfFileExists{microtype.sty}{
  \usepackage[]{microtype}
  \UseMicrotypeSet[protrusion]{basicmath} 
}{}
\makeatletter
\@ifundefined{KOMAClassName}{
  \IfFileExists{parskip.sty}{%
    \usepackage{parskip}
  }{
    \setlength{\parindent}{0pt}
    \setlength{\parskip}{6pt plus 2pt minus 1pt}}
}{
  \KOMAoptions{parskip=half}}
\makeatother
\usepackage{xcolor}
\makeatletter
\ifx\paragraph\undefined\else
  \let\oldparagraph\paragraph
  \renewcommand{\paragraph}{
    \@ifstar
      \xxxParagraphStar
      \xxxParagraphNoStar
  }
  \newcommand{\xxxParagraphStar}[1]{\oldparagraph*{#1}\mbox{}}
  \newcommand{\xxxParagraphNoStar}[1]{\oldparagraph{#1}\mbox{}}
\fi
\ifx\subparagraph\undefined\else
  \let\oldsubparagraph\subparagraph
  \renewcommand{\subparagraph}{
    \@ifstar
      \xxxSubParagraphStar
      \xxxSubParagraphNoStar
  }
  \newcommand{\xxxSubParagraphStar}[1]{\oldsubparagraph*{#1}\mbox{}}
  \newcommand{\xxxSubParagraphNoStar}[1]{\oldsubparagraph{#1}\mbox{}}
\fi
\makeatother

\providecommand{\tightlist}{%
  \setlength{\itemsep}{0pt}\setlength{\parskip}{0pt}}\usepackage{longtable,booktabs,array}
\usepackage{calc} 
\usepackage{etoolbox}
\makeatletter
\patchcmd\longtable{\par}{\if@noskipsec\mbox{}\fi\par}{}{}
\makeatother
\IfFileExists{footnotehyper.sty}{\usepackage{footnotehyper}}{\usepackage{footnote}}
\makesavenoteenv{longtable}
\usepackage{graphicx}
\makeatletter
\def\maxwidth{\ifdim\Gin@nat@width>\linewidth\linewidth\else\Gin@nat@width\fi}
\def\maxheight{\ifdim\Gin@nat@height>\textheight\textheight\else\Gin@nat@height\fi}
\makeatother
\setkeys{Gin}{width=\maxwidth,height=\maxheight,keepaspectratio}
\makeatletter
\def\fps@figure{htbp}
\makeatother

\makeatletter
\@ifpackageloaded{caption}{}{\usepackage{caption}}
\AtBeginDocument{%
\ifdefined\contentsname
  \renewcommand*\contentsname{Table of contents}
\else
  \newcommand\contentsname{Table of contents}
\fi
\ifdefined\listfigurename
  \renewcommand*\listfigurename{List of Figures}
\else
  \newcommand\listfigurename{List of Figures}
\fi
\ifdefined\listtablename
  \renewcommand*\listtablename{List of Tables}
\else
  \newcommand\listtablename{List of Tables}
\fi
\ifdefined\figurename
  \renewcommand*\figurename{Figure}
\else
  \newcommand\figurename{Figure}
\fi
\ifdefined\tablename
  \renewcommand*\tablename{Table}
\else
  \newcommand\tablename{Table}
\fi
}
\@ifpackageloaded{float}{}{\usepackage{float}}
\floatstyle{ruled}
\@ifundefined{c@chapter}{\newfloat{codelisting}{h}{lop}}{\newfloat{codelisting}{h}{lop}[chapter]}
\floatname{codelisting}{Listing}

\makeatother
\makeatletter
\@ifpackageloaded{caption}{}{\usepackage{caption}}
\@ifpackageloaded{subcaption}{}{\usepackage{subcaption}}
\makeatother

\ifLuaTeX
  \usepackage{selnolig}  
\fi
\usepackage[]{natbib}
\usepackage{bookmark}

\IfFileExists{xurl.sty}{\usepackage{xurl}}{} 
\hypersetup{
  pdftitle={Title},
  pdfauthor={Author 1; Author 2},
  pdfkeywords={3 to 6 keywords, that do not appear in the title},
  colorlinks=true,
  linkcolor={blue},
  filecolor={Maroon},
  citecolor={Blue},
  urlcolor={Blue},
  pdfcreator={LaTeX via pandoc}}

\newcommand{\anon}{1}

\usepackage{bm}
\usepackage{amsthm}

\newtheorem{lemma}{Lemma}
\newtheorem{theorem}{Theorem}
\newtheorem{objective}{Objective}

\newtheorem{definition}{Definition}[section]

\newtheorem{proposition}{Proposition}[section]
\newtheorem{remark}{Remark}[section]

\newcommand{\R}{\mathbb{R}}

\newcommand{\E}{\mathbb{E}}
\newcommand{\pr}{\mathbb{P}}

\newcommand{\T}{\mathbf{H}}
\newcommand{\Tinst}{H_{ij}}
\newcommand{\I}[1]{\textbf{1}\{ #1 \}}

\newcommand{\Ei}{\bar{A}_i}
\newcommand{\ei}{\bar{a}_i}

\newcommand{\muTrue}{\mu(\ei; \Xouti, \outCoef_0)}

\newcommand{\Xouti}{\mathbf{X}^{out}_i}

\newcommand{\Xintj}{\mathbf{X}^{int}_j}

\newcommand{\PM}{PM$_{2.5}$\xspace}

\newcommand{\nBeneficiaries}{2,213,471\xspace}

\newcommand{\Wweight}{\nu}

\usepackage{mathalfa}

\usepackage{xspace}

\newcommand{\dual}{\texttt{Percent Poor}\xspace}

\newcommand{\Imb}{\texttt{Disparity}\xspace}

\newcommand{\dualLow}{\texttt{Low Poverty Group}\xspace}
\newcommand{\dualHigh}{\texttt{High Poverty Group}\xspace}

\newcommand{\TotEff}{\texttt{TotalEffect}\xspace}
\newcommand{\EstTotEff}{\widehat{\texttt{TotalEffect}}\xspace}

\newcommand{\SNR}{3\xspace}

\newcommand{\TotPop}{\texttt{TotPop}\xspace}
\newcommand{\TotHeatInp}{\texttt{TotHeatInput}\xspace}
\newcommand{\TotOp}{\texttt{TotOpTime}\xspace}

\newcommand{\bigO}{\mathcal{O}}

\newcommand{\totalCost}{\$48,757,024,000}

\DeclareMathOperator*{\argmin}{arg\,min}

\newcommand{\outCoef}{\pmb{\theta}}
\newcommand{\outCoefb}{\pmb{\alpha}}
\newcommand{\outCoeftx}{\pmb{\beta}}
\newcommand{\intCoef}{\pmb{\gamma}}

\newcommand{\XJ}{\mathbf{\mathbb{X}}_{1:J}^{int}}
\newcommand{\Xn}{\mathbf{\mathbb{X}}_{1:n}^{out}}
\newcommand{\Xj}{\Xintj}

\newcommand{\etaMap}{\eta(\Xn, \T_j)}
\newcommand{\etaMapS}{\eta(\Xn, \T_j)}
\newcommand{\B}{\mathcal{B}}

\newcommand{\al}{A-Learning\xspace}

\newcommand{\policy}{\pi}

\usepackage{enumitem}
\newcommand\numberthis{\addtocounter{equation}{1}\tag{\theequation}}

\usepackage[capitalise]{cleveref}

\usepackage{xr}
\usepackage{authblk}
\usepackage{setspace}
\providecommand{\tightlist}{%
\setlength{\itemsep}{0pt}\setlength{\parskip}{0pt}}

\begin{document}

\def\spacingset#1{\renewcommand{\baselinestretch}%
{#1}\small\normalsize} \spacingset{1}


\renewcommand\Authfont{\fontsize{12}{14.4}\selectfont} 
\renewcommand\Affilfont{\fontsize{11}{10.8}\itshape} 

\if1\anon
{
  \title{\vspace{-2em} \spacingset{1} \Large Fair Policy Learning under Bipartite Network Interference: Learning Fair and Cost-Effective Environmental Policies}
  \date{}
  \author[1]{Raphael C. Kim}
  \author[1]{Rachel C. Nethery}
  \author[1]{Kevin L. Chen}
  \author[2,1]{Falco J. Bargagli-Stoffi}
  \affil[1]{Department of Biostatistics, Harvard T.H. Chan School of Public Health (Boston, MA)}
  \affil[2]{Department of Biostatistics, UCLA Fielding School of Public Health (Los Angeles, CA)}
  \maketitle
} \fi

\if0\anon
{
  \bigskip
  \bigskip
  \bigskip
  \begin{center}
      {\LARGE\bf Fair Policy Learning under Bipartite Network Interference: Learning Fair and Cost-Effective Environmental Policies}
\end{center}
  \medskip
} \fi

\bigskip
\begin{abstract}
Numerous studies have shown the harmful effects of airborne pollutants on human health. Vulnerable groups and communities often bear a disproportionately larger health burden due to exposure to airborne pollutants. Thus, there is a need to design policies that effectively reduce the public health burdens while ensuring cost-effective policy interventions. Designing policies that optimally benefit the population while ensuring equity between groups under cost constraints is a challenging statistical and causal inference problem. In the context of environmental policy this is further complicated by the fact that interventions target emission sources but health impacts occur in potentially distant communities due to atmospheric pollutant transport---a setting known as \textit{bipartite network interference} (BNI). To address these issues, we propose a fair policy learning approach under BNI. Our approach allows to learn cost-effective policies under fairness constraints even accounting for complex BNI data structures. We derive asymptotic properties and demonstrate finite sample performance via Monte Carlo simulations. Finally, we apply the proposed method to a real-world dataset linking power plant scrubber installations to Medicare health records for more than 2 million individuals in the U.S. Our method determine fair scrubber allocations to reduce mortality under fairness and cost constraints.
\end{abstract}

\noindent%
{\it Keywords: Causal Inference; Interference; Fairness; Pareto Optimality; Cost-Effectiveness; Environmental Health; Air Pollution} 
\vfill

\newpage
\spacingset{1.8} 

 \section{Introduction}\label{sec:intro}

\subsection{Motivation}

Airborne pollutants, most notably fine particulate matter of an aerodynamic diameter smaller than 2.5$\mu m$ (\PM), pose a major risk to human health. Many studies have demonstrated the significant relationship between short-term and long-term exposure to fine \PM exposure and heightened risks of mortality and morbidity \citep{Nethery2020,Wu2020,Henneman2023}. Furthermore, there is evidence that specific individuals and communities---including low-income, older, and certain minority groups---bear disproportionate burdens from exposure to \PM\citep{bargaglistofficre,jbaily2022air,Josey2023}.

Coal-fired power plants are the largest source of sulfur dioxide emissions in the U.S., which is a major contributor to secondary \PM formation \citep{Massetti2017}. A recent study estimated that coal-fired power plant emissions were related to 460,000 additional deaths in the U.S. from 1999 to 2020 \citep{Henneman2023}. Using 2024 value of statistical life estimates, the economic impact of such additional deaths might amount to approximately \$300 billion per year \citep{kearsley2024hhs} which is roughly 1\% of US GDP \citep{bea_gdp}. 

An effective intervention strategy for reducing \PM concentrations is installing flue gas desulfurization (FGD) equipment (so called ``scrubbers'') on coal-fired power plants \citep{ZiglerStatSci2021}. Scrubbers, in fact, have proven to be an effective solution, removing at least 90\% of sulfur dioxide emissions from power plants on which they are applied \citep{Srivastava2001} and, consequently, reducing the health burdens associated with exposure to power plants emissions.

Given the astronomical impact of \PM on health and the economy, and its disproportionate burden on certain vulnerable groups, it is of paramount importance to design fair and cost-effective environmental policies aimed at protecting human health by reducing \PM exposure \citep{usepa2022a}. 

This research examines how to develop treatment allocation policies that balance fairness with Pareto efficiency in the context of environmental policies. We explore targeting strategies where policymakers seek allocations that achieve Pareto optimality: that is, no alternative policy could enhance outcomes for one sensitive group without diminishing outcomes for another. From this Pareto-efficient frontier, the decision-maker selects the allocation that best satisfies fairness criteria. This framework applies to health and social welfare program design and draws inspiration from the medical principle of ``\textit{primum non nocere}'' (first do no harm) \citep{rotblat1999hippocratic}.

To inform the design of such Pareto optimal cost-effective policies, we need statistical tools that can provide reliable estimates of health benefits and determine where interventions should be targeted according to Pareto efficiency.
In this paper, we develop a method to address this challenge.

\subsection{Related Work}

Statistical methods used for studying the impacts of large-scale environmental policies aimed at curbing emissions must address two problem features that cannot be accommodated in conventional causal inference methods. First, intervention strategies are implemented on coal-fired power plants (which we will refer to as ``intervention units''), but health impacts are measured in surrounding communities (which we will refer to as ``outcome units''). This creates a so-called \textit{bipartite network} data structure. Second, the complex process in which airborne pollutants react in the atmosphere and are transported (e.g., by the wind) means that intervening at a single power plant can potentially affect health in many distant communities, resulting in a network of connections between power plants and communities. 

In causal inference, the presence of such connections creates \textit{interference}, where the effects of one unit's treatment can spill to (possibly many) other units and can affect their potential outcomes \citep{cox1958planning, hudgens2008toward, dominici2021controlled}. The aggregate of these two features, bipartite network data and interference, has been referred to as \textit{bipartite network interference} (BNI) \citep{ZiglerStatSci2021}.

The BNI literature has primarily focused on causal estimation and inference, spanning methods for average and heterogeneous treatment effects in both cross-sectional and longitudinal settings \citep{PougetAbadieNeurips2019,doudchenko2020causal,ZiglerStatSci2021,chen2024environmental,chen2024differenceindifferences,song2024bipartite,ZiglerForastiere25Bsts}. Several studies have used these methods to assess the effectiveness of environmental policies, particularly in the context of air quality interventions \citep{ZiglerStatSci2021, chen2024environmental, chen2024differenceindifferences, song2024bipartite,ZiglerForastiere25Bsts}. However, while these methods excel at assessing whether these policies work, they offer no guidance on targeting; that is, they do not inform us about who should receive treatment to maximize the policy impact, or ensure fairness.

Policy learning methods, while well-developed in non-interference \citep{Schulte_2014,KosorokLaber19} and standard (non-BNI) interference settings \citep{suModelingEstimation2019, viviano2024policy, ParkMRTP24,zhang2025individualizedpolicyevaluationlearning}, have received limited attention in BNI contexts. 
Furthermore, the majority of the fairness literature has primarily focused on non-interference settings \citep{nabi2018fair, Nabi2019, kimZubi23, frauen2023fairoffpolicylearningobservational, viviano2024fair}, with limited work in the interference setting \citep{YangCuiFairInt24}.
A particular challenge for adapting these policy learning methods to BNI settings is that interventions are non-separable across subpopulations --- that is, a single treatment decision can simultaneously impact multiple subgroups. As a result, optimal or fair policy learning must balance their objectives against this constraint.
To our knowledge, only \cite{kim2024optimalenvironmentalpoliciespolicy} has explored the development of optimal policy methods under BNI with cost constraints, but they did not consider fairness objectives. Critically, existing methods are not tailored to answer the question: ``\textit{How should policies on bipartite networks be designed to maximize both cost-effectiveness and equity?}'' However, this consideration is essential for both medical decision-making and environmental policy, where interventions often have disparate impacts on vulnerable communities.

\subsection{Contribution and Organization}

To account for the shortcomings detailed in the previous sections and answer the policy-relevant question, we propose a novel fair policy learning method for BNI settings. To the best of our knowledge, we are the first to consider fair policy learning under BNI. 
Our approach optimally balances welfare gains across subgroups by ensuring policies achieve \textit{Pareto optimality} under the \textit{first do no harm} principle. 
This framework offers important advantages over typical counterfactual fairness approaches, which do not inherently guarantee such welfare protections.

We derive the theoretical optimality conditions for our method, demonstrating that our method yields a solution on the Pareto frontier with high probability, and achieves $\bigO_p(n^{-\frac{1}{2}})$ regret bounds.
Then, we test our method in Monte Carlo simulations to assess its empirical performance.
Finally, we demonstrate its application using a rich dataset of Medicare claims from \nBeneficiaries beneficiaries, combined with power plant locations and characteristics, pollution transport networks, community demographics, and scrubber installation costs. 

The goal of our application is to learn power plant scrubber installation policies that maximize fairness in health benefits across communities with high vs. low rates of Medicaid eligibility (a proxy for poverty) under cost constraints. Ensuring that a policy is fair to the Medicaid eligible (high poverty) subgroup might be of particular interest because impoverished individuals are known to be more susceptible to pollution-related health harms \citep{Josey2023}. Moreover, because the costs of healthcare for Medicaid enrollees are borne by federal and state governments, policymakers may have an economic interest in ensuring that environmental policies provide sufficient protection of health in this group.

In fact, our approach ensures that (i) no population subgroup gains with another subgroup being harmed, and (ii) unfairness between subgroups is minimized subject to a cost-constraint (cost-effectiveness). We evaluate the performance of our proposed methods under various budgetary and welfare constraints, and examine the extent of unfairness in the analyzed scrubber installation policy. This analysis provides insights into the trade-offs required to achieve fair policies, including the costs involved and the corresponding reductions in mortality rates.

The rest of the paper is organized as follows: \cref{sec:setup} formally introduces the mathematical setup, the objective and identification. \cref{sec:method} introduces our method and its theoretical properties. In \cref{sec:simulations}, we investigate the finite-sample performance of our proposed method through Monte Carlo simulations. In \cref{sec:rwd}, we present the details and results of our motivating application on the scrubber installation policy analysis. Finally, we end with a discussion of the proposed method and its limitations in \cref{sec:conclusions}.
\section{Setup}\label{sec:setup}

Suppose we have $J$ intervention units indexed by $j \in [J]$ and $n$ outcome units indexed by $i \in [n]$. In our motivating application, intervention units correspond to power plants while outcome units correspond to communities, as defined by ZIP codes. Let $Y_i$ denote the observed outcome for outcome unit $i$---e.g., the mortality rates at the ZIP code level.

Let $\Xouti \in \mathbb{R}^p$ denote the vector of covariates for outcome unit $i$, and $\Xintj \in \mathbb{R}^q$ denote the vector of covariates for intervention unit $j$. $\XJ$ denotes the covariates for all intervention units or $\XJ=\{ \Xintj \}_{j \in [J]}$. $\T \in \mathbb{R}^{n\times J}$ or $\{ \T_{ij} \}$ denotes the `\textit{interference map}' or `\textit{bipartite adjacency matrix}'. 
The elements of $\T$ correspond to the strength of pollution transport from a particular intervention unit $j$ to each outcome unit $i$---e.g., $\T_{ij}$ represents the strength of connections between outcome unit $i$ and intervention unit $j$. Thus, $\T_i=(\T_{i1} \dots \T_{ij} \dots \T_{iJ})$ denotes a row of the interference map which represents the strength of the connection between outcome unit $i$ and all $J$ intervention units. Conversely, $\T_j^\top=(\T_{1j} \dots \T_{ij} \dots \T_{nj})^\top$ denotes the column of the interference map which captures the strength of connections between intervention unit $j$ and all $n$ outcome units. 
In our motivating application, $\T$ is derived from meteorological variables (primarily wind patterns) obtained from the National Oceanic and Atmospheric Administration (NOAA) Air Resources Laboratory and aggregated over space and time. The construction of these measures is described in detail by \cite{Henneman2019}, and has previously been used to evaluate the causal effects of environmental policies \citep{chen2024environmental, kim2024optimalenvironmentalpoliciespolicy}.

$\mathbf{A}$ denotes the treatment status (or intervention status) vector. In particular, $\mathbf{A}=(A_1 \dots A_j \dots A_J) \in \{ 0, 1 \}^J$, where 1 corresponds to treating a particular power plant and 0 corresponds to not treating. In our application, treatments denote whether or not a scrubber is installed at a power plant. Outcomes will depend on the so-called exposure mapping \citep{AronowSamii2017, bargagli2025heterogeneous}, which maps the vector of intervention unit-level treatments to outcome units. In our work, we use the following exposure mapping: $\Ei=\frac{1}{J} \sum_{j=1}^J \Tinst A_j$; that is, the linear combination of treatments at the intervention unit level weighted by the strength of the connection between the outcome unit $i$ and intervention unit $j$ (interference map). 
Let $Y_i(\ei)$ denote the potential outcome for unit $i$ under exposure level $\ei$. Note that, in this work, we will consider smaller values of $Y_i$ as desirable---e.g., a smaller mortality rate at the ZIP code level. 

We aim to ensure fairness across subgroups defined by an outcome covariate $S$. Specifically, we treat the subgroup indicator $S \in \{0,1\}$ as an additional outcome-level covariate. In other words, $S_i$ and $\Xouti$ form the complete set of outcome-level covariates. It will be convenient to let $\Xn$ denote the covariates for all outcome units, or $\{ (\Xouti, S_i)_{i} \}_{i \in [n]}$.

Let the policy function be given by $\pi: \Xintj \times \eta(\Xn, \T_j) \mapsto [0,1]$, where $\eta$ summarizes the covariates from all $(S_i, \Xouti)$ into a vector in $\R^{p +1}$ using $\T_j$; see \cref{remark:policyFx} for more on this formulation. Additionally define $e_j$ as the propensity score for intervention unit $j$, or $e_j=\mathbb{P}(A_j=1 \mid \Xj)$.

Notationally, $\E$ will denote the mean with respect to any variables not explicitly conditioned on---including $\Xouti, \Xintj$ and $Y(\ei)$---and $\E_\pi$ will mean the expectation taken under policy $\pi$. $\E_n$ denotes the empirical mean as in the empirical process literature \citep{kosorok2008introduction}. $\mathcal{N}(\varepsilon, \mathcal{F}, L^2)$ will denote the covering number of $\mathcal{F}$ at scale $\varepsilon$ and metric $L^2$ \citep{kosorok2008introduction}. $\lesssim$ will mean $\leq$ up to a finite, positive constant. 

\begin{remark}\label{remark:policyFx}
    Above, we have considered the policy function as a mapping from (i) the intervention covariates and (ii) some functional of our outcome unit measures and intervention map, to a probability. This is different than non-bipartite interference settings since the policy decision about whether or not we treat an intervention unit is governed by more than just intervention level covariates. The choice of $\eta$ depends on the scientific application. For example, in our study of the effects of air pollution, we might take $\eta$ to be an $\T_j$ weighted summary of our outcome covariates, $\eta(\Xn, \T_j)=\frac{1}{n} \sum_{i=1}^n \Tinst (S_i, \Xouti)$. Here, $\eta$ averages outcome covariates using $\T_j$ to up or down-weight the covariate value according to how much power plant $j$ affects a given ZIP code $i$.  Such approaches for summarizing group-level covariates are common in interference literature (e.g. \cite{ParkMRTP24, kim2024optimalenvironmentalpoliciespolicy})
\end{remark}

\begin{remark}
    The policy function $\policy$ depends on intervention unit $j$, and $\policy$  does not vary with a particular value of $S_i$. This is because we are in the bipartite network interference setting, where treating an intervention unit affects potentially many outcome units. Thus, a policy decision will affect all subgroups $s \in S$, and a policy learner must optimize whichever objectives of interest under this setup.
\end{remark}

\subsection{Identification}

We begin by demonstrating causal identification. 
Following \cite{kim2024optimalenvironmentalpoliciespolicy}, we make the following identifying assumptions:
\begin{enumerate}[label={\bfseries (Id\arabic*)}]
\tightlist
    \item \textit{ Consistency of outcome units}: \label{ass:consistency} $Y_i=Y_i(\ei)$.
    \item \textit{ Positivity}: \label{ass:positivity} 
    $P(\Ei=\ei \mid \Xouti, \XJ, \T_i) > 0$ for all $\ei$.
    \item \textit{ Unconfoundedness}: \label{ass:noConfounding}
    $Y_i(\ei) \perp \ei \mid \Xouti, \XJ, \T_i$.
    \item \textit{ Intervention covariates are independent of potential outcomes given the exposure mapping $\ei$ and outcome covariates}: \label{ass:interferenceMapRichness} $Y_i(\ei) \perp \XJ \mid \Xouti, \T_i, \ei $.
\end{enumerate}

\labelcref{ass:consistency}-\labelcref{ass:noConfounding} are adaptations of the standard consistency, positivity, and unconfoundedness assumptions in causal inference to our BNI setting. Assumption \labelcref{ass:interferenceMapRichness} encodes domain knowledge from air pollution epidemiology \citep{Henneman2019}: outcome unit potential outcomes are determined by aggregate exposure levels $\Ei$ (air pollution levels) rather than individual treatment values of intervention units $A_j$. Thus, when learning the distribution of $Y_i(\ei)$ given $\Ei$ and $\Xouti$, it is not necessary to condition on $\XJ$.

\cite{kim2024optimalenvironmentalpoliciespolicy} has shown that, under assumptions {\labelcref{ass:consistency}} -- {\labelcref{ass:interferenceMapRichness}}, we have the following identification results:
    \begin{equation}
        \E[Y_i(\ei)]= \E\big[\E[Y_i \mid \Xouti, \T_i, \Ei=\ei]\big].
    \end{equation}

\subsection{Modeling Assumptions}

Motivated by domain knowledge from air pollution transport and epidemiology \citep{Henneman2019}, we make the following outcome model assumption, and summary mapping assumption.

\begin{enumerate}[label=(\textbf{M\arabic*})]
\tightlist
    \item \label{ass:outcomeModel} \textit{Outcome model is linear in the exposure mapping}:
    \begin{align*}
    Y_i(\ei; \outCoef_0) &= \muTrue+\epsilon_i =  f_0(\Xouti, \outCoefb_0)+ \ei \cdot f_A(\Xouti, \outCoeftx_0) +\epsilon_i \numberthis \label{eq:meanModel}
    \end{align*}
    for mean zero independent random variable $\epsilon_i$. Above, $f_0$ captures the main effect of the covariates and $f_A$ captures the (heterogeneous) treatment effect. 
    \item \label{ass:summaryFxal} \textit{Summary functional is linear in $\Xouti$.}

    We assume that the summary mapping $\etaMap$ is linear in $S_i, \Xouti$.

\end{enumerate}

We comment on these two assumptions. 
Regarding \labelcref{ass:outcomeModel}, we first remark that the forms which $f_0$ and $f_A$ take can be general. For example, these can range from standard parametric models, to various nonparametric models such as tree ensembles; we will call for the treatment effect function class to be Donsker (see \labelcref{ass:donskerOutcome}). 
Second, assuming that the outcome model is additive in treatment exposure with neighbor heterogeneity is a common assumption made in the interference literature to ensure tractable causal inference and policy targeting. Such a model has been considered closely in bipartite literature \citep{PougetAbadieNeurips2019,doudchenko2020causal,kim2024optimalenvironmentalpoliciespolicy}, for causal estimation and policy learning. In the non-BNI literature, similar models have been considered \citep{ZhaoSmallErt21, Liu16, Liu19, ParkKang22, zhang2025individualizedpolicyevaluationlearning}. For example, the linear-in-means models is considered in \cite{Liu16, Liu19} and \cite{ParkKang22} and a more general heterogeneous, additive model in \cite{zhang2025individualizedpolicyevaluationlearning}. Our model can be thought of as a bipartite analog to the additive heterogeneous model in \cite{zhang2025individualizedpolicyevaluationlearning}, with additional heterogeneity captured by the interference mapping $\T$.

We contrast these models to alternatives considered in policy learning and causal modeling under interference. Generally, one would either assume an additive-in-treatment outcome model, as above, or \textit{partial interference}. Under partial interference, interference is restricted up to some finite number of neighbors (e.g. \cite{Hudgens2008,Tchetgen2010,viviano2024policy,ParkMRTP24,bargagli2025heterogeneous}). Such an assumption often involves simplifications to the treatment structure in order to model the interactions between neighbors and create estimators with valid inferential properties. Researchers will typically assume some combination of 
anonymous interference (e.g. \cite{Hudgens2008,Tchetgen2010}) and/or potentially subjective definitions of direct and spillover effects (e.g. \cite{ZiglerStatSci2021, Qu2022}) to reduce the complexity of the model. 
For causal estimation, these assumptions will allow for more general outcome models to be utilized. However, for policy targeting, anonymous interference assumptions and spillover effects limit the ability to individualize treatment, yielding policies that specify how treatment should be given to a group on average \cite{ParkMRTP24}. Interestingly, \cite{viviano2024policy} provide methods for individualized policy learning under interference using a particular experimental setup, which calls for randomly subsetting the population in a single network in order to extrapolate effects. This experimental setting, however, is not relevant to this study.

In contrast to partial interference, we require the ability to perform policy targeting to a subset of communities under BNI, with varying strengths of node-to-node interference and dependence across space from air pollution transport --- see \cref{sec:SpatialSetup} for our precise setup. Accounting for (i) long-range spillovers, (ii) location-based heterogeneity, and (iii) policy targeting understanding the effect on certain subgroups, makes our theoretical setup unsuited for partial interference, anonymity, and direct/spillover effects setups. As a result, we adopt this realistic yet expressive assumption to encompass our unique interference structure, and ensure tractable policy targeting of various subgroups.

We finally remark that our framework is not strictly dependent on \labelcref{ass:summaryFxal}. Modeling more complex functionals simply require that we can estimate said functional well in a spatially mixing setting. We leave these generalizations to future work. For the purposes of our scientific application, these working assumptions are interpretable, in line with our domain setting and have already been adopted in the literature (e.g. \cite{ParkMRTP24, kim2024optimalenvironmentalpoliciespolicy} utilize averages of neighbor covariates as inputs for interference models, among others).

\section{Fair Policy Learning under BNI}\label{sec:method}

We will now describe our methodology for fair policy learning under BNI.

In order to define our fair policy learning rigorously, we first need to quantify the benefit that a particular subgroup $s \in S$ experiences under some policy $\pi$. We do this using the so-called ``\textit{Welfare Function}'' \citep{Manski04,KitagawaTetenov18}. Let the treatment vector of intervention units excluding power plant $j$ be denoted by $\bm{A}^{(-j)}$. Further define the exposure mapping without intervention unit $j$ treated, under $\pmb{A}^{(-j)}$, by $\ei^{(-j)}(\pmb{A}^{(-j)})$, where $\ei^{(-j)}(\pmb{A}^{(-j)})=\frac{1}{J} \sum_{j' \neq j}^J H_{ij'} A_{j'}$.

\begin{definition}{\bf Welfare Function (under BNI).}
The welfare function under policy $\policy$ for subgroup $s$ is given by:
\vspace{-2em}
\begin{align*}
    W_s(\pi) &= \E[ (Y_i(\ei^{(-j)}+\Delta_{ij}) -Y_i(\ei^{(-j)})) \cdot \pi^{1:J}(A_j=1,\bm{A}^{(-j)} \mid \{ \Xn, \T \} ) \mid S_i=s] \numberthis \label{eq:ogWelfare}
\end{align*} 
\end{definition}
\vspace{-2em}
where $\pi^{1:J}$ denotes the policy function for all $J$ units and $\Delta_{ij}$, is defined as an \textit{incremental exposure}---that is, the exposure change for outcome unit $i$ when treating intervention unit $j$ vs. not. In our case, $\Delta_{ij}$ is equal to $\Tinst/J$ (as our exposure mapping $\ei$ normalizes by $J$). More generally, $\Delta$ could assume different values depending on scenario-specific definitions of the incremental exposure. In words,~$W_s(\cdot)$ measures the average difference in potential outcomes when treating intervention unit $j$ (vs. not) across all other treatment settings ($\bm{A}^{(-j)}$), among subgroup $s$. Again, note that the smaller the outcome, the better (e.g. less hospitalizations), and the same applies to the welfare function.

The objective of this paper is to develop a policy learning method that minimizes the welfare imbalance between subgroups subject to Pareto optimality (thus ensuring we do no harm). At a high level, we are interested in 
\vspace{-1.5em}
    \begin{align*}
         \argmin_{\pi\in \Pi}& \quad |W_{1}(\pi)-W_{0}(\pi)|  \mbox{    s.t.    } \pi \mbox{ is Pareto-optimal }\wedge\mbox{ } \sum_{j=1}^J \pi_j c_j \leq C.
         \vspace{-2em}
    \end{align*}
The latter constraint encodes an overall cost limit for the treatment policy. This constraint ensures the cost-effectiveness of the policy learning objective. To rigorously define this objective, we need to characterize the Pareto Frontier.

First, we start by the following simplification of the welfare function:
\vspace{-2em}
\begin{align*}
    W_s(\pi) &= \E_{\pmb{A}^{(-j)}, \Xouti, \Xintj, Y}[ \Tinst  f_A(\Xouti, \outCoeftx_0) \cdot \pi^{1:J}(A_j=1,\bm{A}^{(-j)} \mid \{ \Xn, \T \} ) \mid S_i=s] \because \ref{ass:outcomeModel} \\
    &= \E_{\Xouti, \Xintj, Y}[ \Tinst f_A(\Xouti, \outCoeftx_0)\cdot \sum_{\pmb{A}^{(-j)}}  \pi(A_j=1, \bm{A}^{(-j)} \mid \{ \Xn, \T \} ) \mid S_i=s] \\
    &=\E_{\Xouti, \Xintj, Y}[ \Tinst f_A(\Xouti, \outCoeftx_0) \cdot \pi(A_j=1  \mid \etaMapS)  \mid S_i=s] \\
    &=\E_{\Xintj, Y}[ \frac{1}{n} \sum_{i=1}^n \Tinst f_A(\Xouti, \outCoeftx_0) \cdot \pi(A_j=1  \mid \etaMapS)  \mid S_i=s] \numberthis 
\end{align*} 
Above, we have shown that with \labelcref{ass:outcomeModel}, $W_s(\pi)$ is linear in $\pi$, independent of the other treatment decisions.
With this simplified welfare function in hand, we can characterize the Pareto Frontier using a standard adaptation from \cite{Negishi1960}.

\begin{proposition}{Pareto Frontier.}\label{thm:paretoFrontNeg}
The Pareto Frontier $\Pi_0$ is given by:
\vspace{-1.5em}
    \begin{equation}
    \Pi_{0}=\{ \pi : \arg \inf_{\pi \in \Pi} \sum_{s \in \{ 0, 1 \}} \Wweight_{s} W_{s}(\pi), \pmb{\Wweight}=(\Wweight_0, \Wweight_1) \in \Delta^{2}  \}
\end{equation}
where $\Delta^l$ denotes the $l$ simplex. The proof is outlined in \cref{thm:ParetoFront}.
\end{proposition}
Informally, deviating from a policy on this set will not strictly improve welfare for one of the subgroups (e.g., $S=s$) without decreasing welfare on any of the other subgroups (e.g., $S=s'$), in line with the `first do no harm' principle.

This leads to the concrete objective problem of minimizing the welfare difference while being constrained to lying on the Pareto frontier:
\begin{objective}[Fair Policy Learning Under Cost Constraint]
\begin{align*}
 \min_{\pi \in \Pi} \quad & |W_1(\pi)-W_0(\pi)| \\
    \mbox{ s.t. } & \sum_{s \in \{ 0, 1 \}} \Wweight_{s} W_{s}(\pi) \leq \inf_{\pi \in \Pi} \sum_{s \in \{ 0, 1 \}} \Wweight_{s} W_{s}(\pi): \pmb{\Wweight} \in \Delta^{2} \\
    & \sum_{j=1}^J \pi_j c_j \leq C.
\end{align*}
\end{objective}

\subsection{Optimization via Quadratic Programming}

To solve this problem, we need an approximation of the Pareto frontier. To achieve that, we require an approximation to the welfare function, and a discretization of the frontier. First, let the estimator to the welfare function under policy $\policy$ and subgroup $s$ be given by
\begin{align*}
    \hat{W}_{s}(\pi) &= \frac{1}{J} \sum_{j=1}^J \pi(A_j=1 \mid \etaMap) \cdot \EstTotEff_{j}(s)
\end{align*}
where the treatment effect for treating intervention unit $j$ on outcome units with $S_i=s$ is estimated as 
$$ \EstTotEff_{j}(s)=\frac{1}{n}\sum_{i=1}^n \frac{\I{S_i=s}}{p_s}\Tinst f_A(\Xouti, \hat{\outCoeftx}) $$
Above, $p_s = \E[S_i=s]$ (which needs to be estimated) and $\hat{\outCoeftx}$ is estimated using the \al estimating equation proposed in \cite{kim2024optimalenvironmentalpoliciespolicy}, which was shown to be $\sqrt{n}$ consistently estimated provided either the baseline model $f_0$ or propensity score model $e$ is correctly specified.

Now, we consider the discretization of the grid. Suppose we have an equally-spaced grid of length $K$, defined by $\Wweight_{k,0} \in (0,1)$ for $k \in [K]$, with $\Wweight_{k,1}=1-\Wweight_{k,0}$.
Then, we define the approximate Pareto frontier as follows:
\vspace{-2em}
\begin{equation}
    \hat{\Pi}_{0} = \bigg\{ \pi \in \Pi, k \in [K]: \pi \in \arg \inf_{\pi \in \Pi} \sum_{s \in \{ 0, 1 \}} \Wweight_{k,s} \hat{W}_{s}(\pi) \bigg\}
\end{equation}
We consider these constraints up to slack $\frac{\lambda}{K}$, for $\lambda \in \R^+$.  Define 
\begin{equation}\label{eq:feasibleApproxPareto}
    \hat{\Pi}_{0}(\lambda) = \{ \pi \in \Pi: \exists k  \quad s.t. \sum_{s \in \{ 0, 1 \}} \Wweight_{k,s} \hat{W}_{s}(\pi) \leq \overline{W}_{k} + \frac{\lambda}{K} \}
\end{equation}
where $\overline{W}_{k}$ is the optimal objective value for $\pmb{\Wweight}_k$, or $\overline{W}_{k}=\inf_{\pi \in \Pi} \sum_{s \in \{ 0, 1 \}} \Wweight_{k,s} \hat{W}_{s}(\pi)$.
We now propose a quadratic program to produce fair policy estimates using the results above. To impose the Pareto optimality, we introduce the auxiliary variable $\mathbf{u}=(u_1 \dots u_K) \in \{ 0, 1 \}^K$. We will let $u_k = 1$ whenever the constraint in \cref{eq:feasibleApproxPareto} holds for $\pi$ and some $a_k$.  For shorthand, let $\pi_j = \pi(A_j=1 \mid \etaMap)$.

We then propose solving: 
\vspace{-2em}
\begin{subequations}\label{eq:optProb}
\begin{align}
 \arg \min_{\pi \in \Pi} \min_{\bm{u}} & \quad |W_1(\pi)-W_0(\pi)| \\
 \mbox{ s.t. } & u_k \sum_{j=1}^J \sum_{s \in \{ 0, 1 \}} \Wweight_{k,s} \EstTotEff_{j}(s) \pi_j \leq u_k \overline{W}_{k} + \frac{\lambda}{\sqrt{n}}: \pmb{\Wweight}_k \in \Delta^{2} \tag{A} \label{con:paretoCondition:a} \\
 & u_k \in \{ 0, 1 \}, \forall k \tag{B} \label{con:paretoMin:b} \\
 & \sum_{k=1}^K u_k \geq 1  \tag{C}  \label{con:onePareto:c} \\
 & \sum_{j=1}^J \pi_j c_j \leq C  \tag{D} \label{con:costConstraints:E}
\end{align}
\end{subequations}

\cref{con:paretoCondition:a}--\cref{con:onePareto:c} are the Pareto optimality constraints, calling for approximate Pareto optimality to be satisfied for at least one of the grid values. If we want to incorporate cost-constraints for known treatment costs $c_j$ and total budget $C$, one can impose a budget constraint in \cref{con:costConstraints:E}.

\subsection{Theoretical Results}\label{sec:theory}

We now show that the procedure achieves a $\sqrt{n}$ regret bound for $K = \sqrt{n}$. Providing theoretical guarantees is non-trivial in our setup because (i) $\etaMapS$ depends functionally on all outcome units, and (ii) concentration of $\EstTotEff_j$ is unclear under BNI.

We address (i) via \cref{lemma:welfarePopEta}, which establishes that we may consider the deterministic population limit of $\etaMap$, thereby eliminating its dependence on individual outcomes.
We use this to show (ii) via \cref{lemma:TE_Est_Concentration} demonstrating that our treatment effect estimator concentrates at a $\sqrt{n}$ rate. We then conclude our main results: our Pareto frontier is estimated well (\cref{thm:ParetoFrontEst}), the solutions on our front are correct with high probability (\cref{thm:ParetoFrontSupp}), and our final regret bound (\cref{thm:RegBd}).

Throughout, we rely on spatial mixing assumptions \labelcref{ass:samplingRegime}–\labelcref{ass:mixing}. The setup is similar to \cite{Jenish2009}, and described precisely in \cref{sec:SpatialSetup}. At a high level, these assumptions state that the spatial sampling regime grows in space with weak dependence across distant locations. 
We additionally make the following \textbf{Estimation Assumptions}.

\begin{enumerate}[label={\bfseries (Est\arabic*)}]
\tightlist
    \item \textit{Boundedness}: \label{ass:bounded}
    {Assume the outcome, estimated \& true propensity score model, estimated \& true outcome model ($Y, \hat{e}, e, \hat{\mu},\mu$) are bounded uniformly by $M < \infty$.}
    \item \textit{Donsker Treatment Effect  Function Class}: \label{ass:donskerOutcome}
    Assume $f_A \in \mathcal{F}_A$ is Donsker, or 
    \vspace{-2em}
    $$ \log \mathcal{N}(\varepsilon, \mathcal{F}_A, L^2) \lesssim \varepsilon^{-m}, \quad m < 2$$ 
    \vspace{-4em}
    \item \textit{Donsker Policy Class}: \label{ass:donskerPolicy}
    Assume $\Pi$ is Donsker, or 
    \vspace{-2em}
    $$ \log \mathcal{N}(\varepsilon, \Pi, L^2) \lesssim \varepsilon^{-p'}, \quad p' < 2$$ 
\end{enumerate}
\vspace{-1.5em}
By our outcome model assumption \labelcref{ass:outcomeModel}, we only need Donsker assumptions on $f_A$.
We now show that we can effectively remove the dependence on the $\Xouti$ in the summary functional (which is a function of all $\Xouti$). 
\begin{lemma}{Concentration of Welfare Function to population $\eta$.}\label{lemma:welfarePopEta}
    Assume boundedness \labelcref{ass:bounded}, spatial mixing \labelcref{ass:infiniteSampling}-\labelcref{ass:mixing}, $\eta$ satisfies \labelcref{ass:summaryFxal}, and \labelcref{ass:donskerPolicy}. Then, for the deterministic limit of $\etaMap$, $\eta_j \in \R^{p+1}$,
    \vspace{-1.5em}
    \begin{align*}
    \E|\frac{1}{J} \sum_{j=1}^J & \frac{1}{n}\sum_{i=1}^n & \Tinst f_A(\Xouti, \outCoeftx_0) S_i (\pi(A_j=1\mid \etaMap, \Xintj)-\pi(A_j=1 \mid \eta_j, \Xintj))| \lesssim \frac{M}{\sqrt{n}}
    \end{align*}
\end{lemma}
\vspace{-1.5em}
The proof is found in \cref{pf:welfarePopEta}. This result relies on spatial mixing to guarantee a well-defined limit, and regularity conditions in order to conclude the convergence result.
\\
With this result in hand, we can conclude concentration of our treatment effect functions.
\begin{lemma}{Concentration of Estimated Treatment Effect Function.}\label{lemma:TE_Est_Concentration}
    Assume boundedness \labelcref{ass:bounded}, Donkser treatment effect and policy class \labelcref{ass:donskerOutcome}-\labelcref{ass:donskerPolicy}, spatial mixing \labelcref{ass:infiniteSampling}-\labelcref{ass:mixing}, and linearity of our models \labelcref{ass:outcomeModel}-\labelcref{ass:summaryFxal}. Then,
    \vspace{-1.5em}
    $$ \E \sup_{\outCoeftx} |\frac{1}{n} \sum_{i=1}^n \Tinst f_A(\Xouti, \outCoeftx_0) - \E[\Tinst f_A(\Xouti, \outCoeftx_0)]| \lesssim \frac{M}{\sqrt{n}}$$
\end{lemma}
\vspace{-1.5em}
The proof is found in \cref{pf:TE_Est_Concentration}. This result relies on spatial mixing to guarantee a well-defined limit, and regularity conditions in order to conclude the convergence result. In particular, the blocking technique described in \citep{Bernstein1927, KuznetsovMohri17} is applied to spatial blocks, upon which standard empirical process tools are used to conclude our rates. In contrast to existing literature (e.g., standard concentration bounds), this result provides concentration bounds for empirical sums of Donsker classes in spatial settings, with arbitrary dependence through functionals.

We now use this to show the guarantees on our fair policy learning procedure.  
\begin{theorem}{Pareto Frontier Estimation.}\label{thm:ParetoFrontEst}
    Assume boundedness \labelcref{ass:bounded}, Donkser Treatment Effect and Policy class \labelcref{ass:donskerOutcome}-\labelcref{ass:donskerPolicy}, spatial mixing \labelcref{ass:infiniteSampling}-\labelcref{ass:mixing}, and linearity of our models \labelcref{ass:outcomeModel}-\labelcref{ass:summaryFxal}. Also, let $K = \sqrt{n}$. Then,
    \vspace{-1.5em}
    $$\E [\sup_{\pmb{\Wweight}, \pi} |\sum_{s \in \{ 0, 1 \}} \Wweight_s W_s(\pi) - \inf_{\pmb{\Wweight}_k} \{ \sum_{s \in \{ 0, 1 \}} \Wweight_{ks} \hat{W}_s(\pi) \}|] \lesssim \frac{M+\lambda}{\sqrt{n}} $$
\end{theorem}
\vspace{-1.5em}
The proof is found in \cref{pf:ParetoFrontEst}. The object of interest is bounded with the triangle inequality using concentration of the treatment effect function and regularity of the policy class to obtain $\sqrt{n}$ rates. This result effectively extends Theorem 4.1 of \cite{viviano2024fair} to the BNI setting.

With this finite-sample bound, we can conclude with high probability that our estimated set of Pareto front solutions yield the solutions on the Pareto frontier with appropriately chosen grid size and slack:
\begin{theorem}{Pareto Frontier Support.}\label{thm:ParetoFrontSupp}
    Assume the conditions of \cref{thm:ParetoFrontEst}, and recall that $K = \sqrt{n}$. Then, for any $\gamma \in (0, 1)$, with $\lambda \leq \frac{\gamma}{CM}$ for universal constant $C$, we have
    $ \mathbb{P}(\Pi_0 \subseteq \hat{\Pi}_0(\lambda)) \geq 1 - \gamma$
\end{theorem}
The proof is found in \cref{pf:ParetoFrontSupp}. This result follows by applying Markov's inequality and using  \cref{thm:ParetoFrontEst} to further bound the probability.

\begin{theorem}{Regret Bound.}\label{thm:RegBd}
Assume the setup of \cref{thm:ParetoFrontEst}. Then,
\vspace{-1.5em}
\begin{align*}
        \E |W_1(\hat{\pi}_\lambda)-W_0(\hat{\pi}_\lambda)| - & \inf_{\pi \in \Pi_0} |W_1(\pi)-W_0(\pi)|  | \lesssim \frac{M}{\sqrt{n}} 
\end{align*}
\end{theorem}
\vspace{-1.5em}
The proof is found in \cref{pf:RegBd}. The argument follows by first bounding the welfare deviation on the full policy class by the welfare deviation on the pareto frontier, which is guaranteed with high probability from the previous result (\cref{thm:ParetoFrontSupp}). Then, estimation guarantees on the welfare functions \cref{thm:ParetoFrontEst} allow us to conclude the result. In summary, we extend the fair policy targeting result under the Welfare Disparity measure to high-dependence BNI settings.

\section{Empirically-based Monte Carlo Simulations}\label{sec:simulations}
In this section, we carry out empirically-based Monte Carlo simulations to evaluate the proposed methodology \citep{knaus2021machine}. 

\subsection{Simulation Setup}

The setup is as follows: we choose simulation parameters to ensure the proportion of treated units in the simulated data is within 0.01 of the empirical treatment rate in the real data and to ensure that the simulated mortality rate is within 0.01 of the empirical mortality rate in our data. Then, we run 1,000 simulations, generating new treatments and outcomes in each, under a signal-to-noise ratio of 3. We maintain the $J=459$ power plants and $n=35,036$ ZIP codes, with each subgroup comprising about half of the total ZIP codes ($n_0=17,517, n_1=17,519$). We then run our method and a competing policy learning method on the simulated data and evaluate the learned policies on the following metrics, averaging across the 1000 simulations (full simulation details and parameters are in \cref{sec:sim_details}):
\begin{itemize}
\tightlist
    \item Welfare for Group 0, $W_0(\pi)$;
    \itemsep0em
    \item Welfare for Group 1, $W_1(\pi)$;
    \item Disparity in the policy defined by:
$\Imb(\pi)=|W_0(\pi)-W_1(\pi)|$.
\end{itemize}
Since we assume smaller outcomes are better, a smaller welfare indicates a more desirable result. For $\Imb$, smaller values are preferable.

We compare our fair policy method to existing constrained welfare maximization methods \citep{Manski04, KitagawaTetenov18}. We note that, practically speaking, we are in fact implementing welfare \textit{minimization} since smaller outcomes are preferred in our context. However, for consistency with literature, we continue to refer to this methodology as welfare maximization. Welfare maximization seeks to maximize welfare of a target population.
Under a utilitarian perspective, this means optimizing a weighted combination of the welfares, for $\Wweight \in [0, 1]$: $\Wweight  W_0(\pi) + (1-\Wweight ) W_1(\pi) $.

We compare the performance of the fair policy to the welfare maximization approach as we vary (i) the maximum permitted \Imb while holding the other constraints fixed, and (ii) the budget while holding the other constraints fixed. The purpose of these comparisons is to examine how the welfares are prioritized and how \Imb consequently scales.

The histogram for the ground-truth individual total effects, by subgroup and overall, are shown in \cref{fig:TE_Sim}. The subgroup sizes are equal in this simulation study. In subgroup 0, the majority of scrubber installations are protective. Specifically, $\E [\TotEff_j(0)] = -261.6$ with 99.6\% of intervention units benefiting this subgroup. On the other hand, in subgroup 1, the majority of scrubber installations are harmful. Specifically, $\E [\TotEff_j(1)] = 105.6$ with 3.9\% of intervention units benefiting this subgroup. Overall, the benefits to group 0 are larger in magnitude than the harms to group 1, yielding overall more protective effects: $\E [\TotEff_j] = -156.0$ with 93.7\% of treatments benefiting the full population. Thus, the fair policy learner must balance benefit to group 0 with potential harm to group 1, minimizing \Imb. 

\begin{figure}[H]
\center
   \includegraphics[width=.85\linewidth]{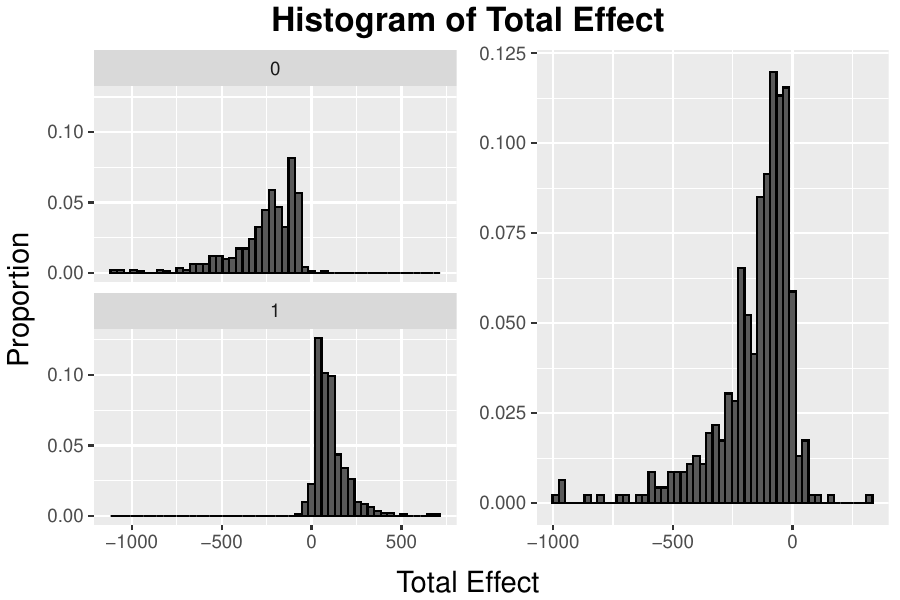}
   \caption{\small Total effect histogram for our simulation setup. These histograms depict the ground-truth individual total effects, by subgroup (left plots) and overall (right plot). The subgroup sizes are equal in this simulation study.\label{fig:TE_Sim}} 
\end{figure}

\subsection{Simulation Results}

The simulation results are shown in \cref{fig:SimResSc1}. As shown in the left column of the plots, the welfare maximization (shown in blue triangles), which prioritizes decreases in welfare in the population overall, leads to treating more units since this benefits group 0 much more and thus the overall population more. However, this yields worse outcomes for subgroup 1 and thus higher \Imb. On the other hand, the fair policy learner (red dots) maintains the welfare at a more consistent level of disparity at the cost of lower total welfare gains, even with greater cost budgets. 

This behavior is observed both for the simulation scenario on the top row, in which we change the maximum permitted \Imb while holding the other constraints fixed, and in the scenario on the bottom row, where we change the budget while holding the other constraints fixed. In particular, we note that the level of \Imb (right column plots) increases steadily for the welfare maximization model both when the admissible disparity increases, but also in the case where larger budgets (\% budget) allow for the treatment of a larger pool of units.

Overall, the simulation results support that the fair policy learner induces lower disparity compared to the welfare maximization method, at the cost of lower total welfare gains. Such a tradeoff is apparent also in the non-interference context \citep{viviano2024policy}. Our simulation study further reveals the challenge of fair policy learning in a BNI setting: we cannot choose to treat only one of the subgroups; the intervention unit will affect both subgroups and we must tradeoff benefits in such a setting while ensuring fairness.


\begin{figure*}[t!]
\centering
\begin{subfigure}[t]{0.49\textwidth}
    \includegraphics[width=\linewidth]{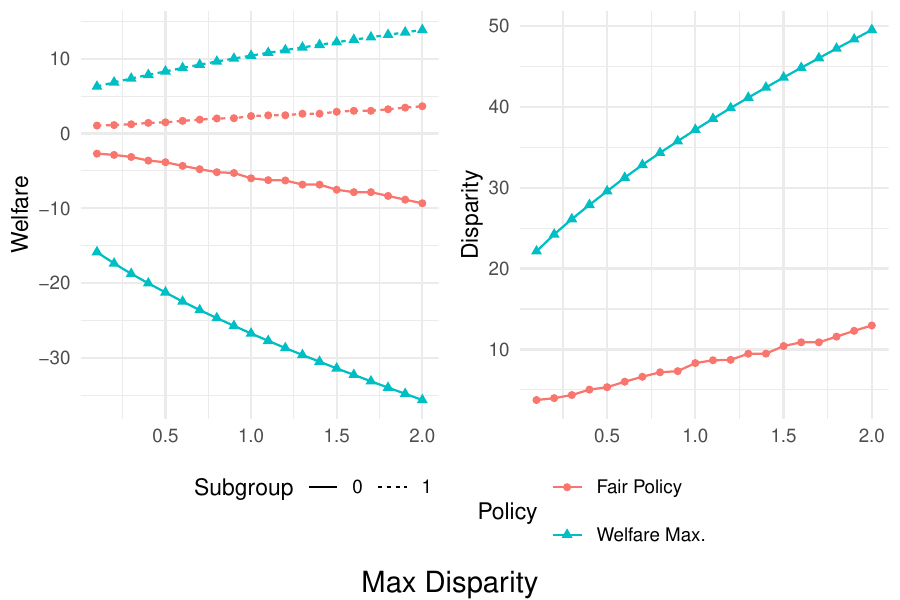}
    \label{fig:consWelfare}
\end{subfigure}
\hfill
\begin{subfigure}[t]{0.49\textwidth}
    \includegraphics[width=\linewidth]{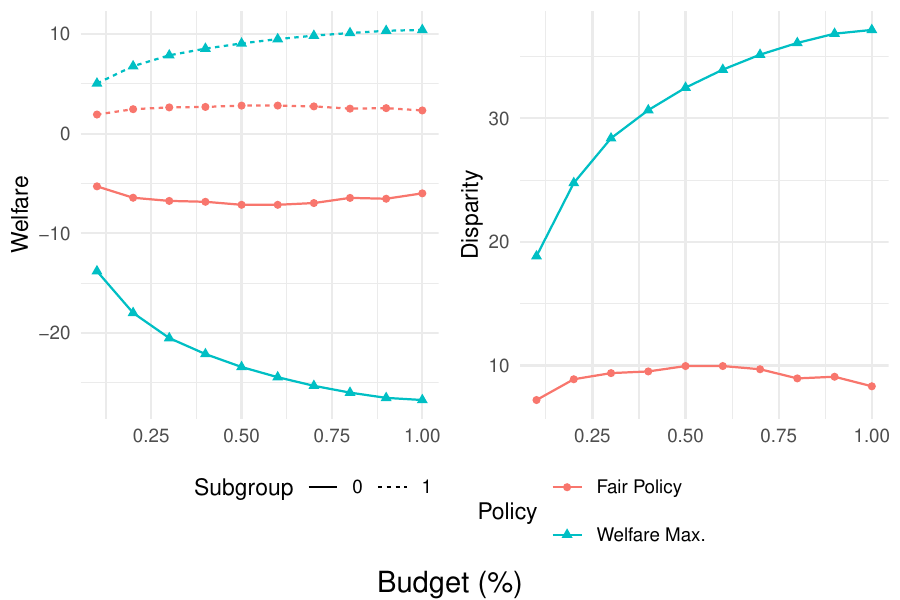}
    \label{fig:consDisparity}
\end{subfigure}    
\caption{\small Two different simulation scenarios: We compare the welfare performance (where lower values correspond to ``better'' outcomes in our case) of the fair policy (blue triangles) to the welfare maximization approach (red dots) as we vary i) the maximum permitted \Imb while holding the other constraints fixed (top row), and (ii) the budget while holding the other constraints fixed (bottom row). Plots in the left column depict results disaggregated by subgroups. Plots in the right column depict aggregate results. \label{fig:SimResSc1}}
\end{figure*}

\section{Fair, Cost-Effective Environmental Policy Analysis}\label{sec:rwd}

We apply the fair policy learning method to our motivating applied problem of learning power plant scrubber installation policies that are fair and cost-effective across socioeconomic subgroups under cost constraints. We begin with our data description.

\subsection{Data Description}\label{sec:data}

Our data consist of three parts: (i) outcome level data with Medicare beneficiary information, (ii) intervention level data with power plant characteristics and scrubber cost information, and (iii) the interference map or characterization of air pollution transport effects between outcome units and intervention units via a reduced complexity atmospheric model known as \textit{HYSPLIT Average Dispersion} (HyADS) \citep{Henneman2019}. We detail each component.

\subsubsection{Outcome Level Data: Medicare Mortality Rates and Covariates}

The mortality outcomes for our analyses are derived from \nBeneficiaries Medicare enrollee records for all Medicare beneficiaries residing in $n=35,036$ ZIP codes across the contiguous U.S. in 2005. Using these records, we compute ZIP code-level 2005 mortality rates, defined as the number of deaths per person-year among these Medicare beneficiaries, which accounts for different amounts of person-time of exposure across ZIP codes.

Several covariates were also measured at the ZIP code level. In particular, U.S. Census socioeconomic and demographic features (from the 2000 decennial census), meteorological \citep{Kalnay1996}, and smoking rate \citep{DwyerLindgren2014} covariate data were obtained. The complete list of ZIP code covariates used in our analyses, and corresponding descriptive statistics, are given in \cref{tab:zip_covs}. We perform log transformations of several covariates (\TotHeatInp, \TotPop, and \TotOp), in order to reduce skew in their distribution, and standardization of covariates.

\subsubsection{Intervention Level Data: Power Plant Data and Interference Map}

Data on scrubber status in the year 2005 and plant characteristics for $J=459$ coal–fired power plants concentrated in the Eastern U.S., which serve as our intervention units of interest, were obtained from the U.S. Environmental Protection Agency (EPA) Air Markets Program Database. We also collect several covariates at the power plant level, which are depicted, along with descriptive statistics, in \cref{tab:pp_covs}. Scrubber installation cost information for power plants that had a scrubber installed was obtained from the Energy Information Administration (EIA) website and was used to build a predictive model to estimate costs of scrubber installation at each power plant without a scrubber, as described in \cite{kim2024optimalenvironmentalpoliciespolicy}. Under this cost model, the estimated total cost for installing a scrubber at all of the 459 power plants is approximately \totalCost. 

To characterize the interference structure or the bipartite adjacency matrix, we utilize the HyADS model. HyADS is a pollution transport model that yields a unit-less metric quantifying the amount of emissions from an individual power plant that were transported (e.g., by wind) to a particular ZIP code \citep{Henneman2019}. HyADS is based on highly accurate meteorological measures from the HYSPLIT model by the National Oceanic and Atmospheric Administration (NOAA) Air Resources Laboratory. These values, which were calculated for all power plant pairs and ZIP codes in the data, form the elements of $\mathbf{H}$.

\subsection{Subgroup Definition}

We construct policies that are fair across subgroups defined based on poverty levels. \dual is defined at the ZIP code level as the percent of Medicare enrollees in a given ZIP code that are also eligible for Medicaid (these individuals are sometimes referred to as ``dual eligible''). The distribution for \dual is shown in \cref{fig:S_Distro_Dual}, with the 0th, 25th, 50th, 75th, and 100th percentiles marked in black dashed lines. 

The histogram shows a long right tail in the poverty levels, with the majority of ZIP codes concentrated at lower poverty levels and a small concentration of ZIP codes at higher poverty levels. In order to capture the group of people living in poverty, we carry out fair policy analysis using subgroups defined as groups of ZIP codes above/below the 75th percentile of \dual in our sample.
We refer to the group above the 75th percentile as \dualHigh and the group below the 75th percentile by \dualLow.

Results from alternative definitions based on the 25th and 50th percentile of the poverty subgroup are shown in Supplementary Material Section \ref{sec:SensDiffSPercs} 
as sensitivity analyses.

\subsection{Estimation Results for Optimal Policies}

We apply the doubly-robust A-learning estimation method from \cite{kim2024optimalenvironmentalpoliciespolicy} to estimate our treatment effects $\TotEff_j, \TotEff_j(0), \TotEff_j(1)$. We show the corresponding breakdown of these quantities in \cref{fig:Hist_TEs_PctPoor}. We find that the vast majority of scrubber installations are protective ($99\%$ of scrubber installations), with the \dualHigh experiencing lower benefits compared to the \dualLow. 
\cref{fig:mapAndHistTEdual} shows the concentration of plants with estimated high health gains from treatment lie in the Midwest and Eastern regions of the U.S.

\subsection{Fair Policy Analysis}

We perform fair policy learning under two policy-making perspectives. To ground our discussion, we will refer to the existing scrubber allocation in 2005 as the \textit{factual} policy.

\begin{enumerate}
\tightlist
    \item \textbf{Clean slate policy analysis}: We perform policy learning as if no plants were treated yet, or from a `clean-slate'. It would not be practical to implement a clean slate policy in our context, as reversing existing treatments and starting from a clean slate would likely be infeasible and potentially unethical. Thus, the purpose of the clean slate analysis is primarily to provide insights into the potential inefficiencies and unfairness of the factual policy and demonstrate improvements offered by fair policy learning.
    \item \textbf{Augmentation policy analysis}: We perform policy learning beginning from the factual policy. This policy augments the factual policy, so we refer to this as the `augmentation' policy analysis. The augmentation policy provides an actionable strategy for expansion upon the factual policy and could inform future policy actions.
\end{enumerate}

Throughout our analyses, we compare the following policy learning methods: (i) fair policy learner, (ii) welfare maximization approach, an adaptation of \cite{KitagawaTetenov18}, (iii) optimal policy learner from \cite{kim2024optimalenvironmentalpoliciespolicy}, (iv) factual policy. Since, to our knowledge, there are no specified policy budget constraints for scrubbers in the U.S., we examine the performance of each policy learner across a range of budget constraints defined as proportions of the universal cost for installing a scrubber at every power plant in the U.S. We examine $(c_n)_{n=1}^{10} \text{ where } c_n = 0.1n$, and additionally include 0.12 since the factual policy satisfies this constraint.

\subsubsection{Clean slate and augmentation policy}

We now detail the results of the \textit{clean slate} and \textit{augmentation policy} analyses in \cref{fig:main_dual}. First, we see that both the clean slate and augmentation policy Analysis benefit the \dualLow (solid lines) more than the \dualHigh (dashed lines). This is in line with \cref{fig:Hist_TEs_PctPoor}. 

Next, we find that the welfare maximization (blue triangles) and the optimal policy (green crosses) have greater welfare gains compared to the fair policy learner (red circles). These policies, however, induce higher disparity compared to the fair policy learner. 

We also note that the factual policy (black square) is higher in the plot compared to the welfare maximization and optimal policy, revealing a relative lack of efficiency in overall welfare benefits compared to optimal decision making policies (the welfare maximizer and optimal policy learner), and lack of fairness compared to the fair policy learner.

Interestingly, under a clean slate policy analysis, no fair, Pareto-optimal policy is learned at 10\% of the cost budget. A higher budget (beginning at 12\%) is required to obtain a fair policy under the clean slate analysis. In contrast, the augmentation policy fails to identify any policy beyond the factual within the 10–20\% budget range. Only at 30\% can an augmented fair policy be found. This indicates that the factual policy exhibits persistent \Imb, and achieving fairness requires a larger budget.



\begin{figure*}[t!]
\centering
\begin{subfigure}[t]{0.49\textwidth}
    \includegraphics[width=\linewidth]{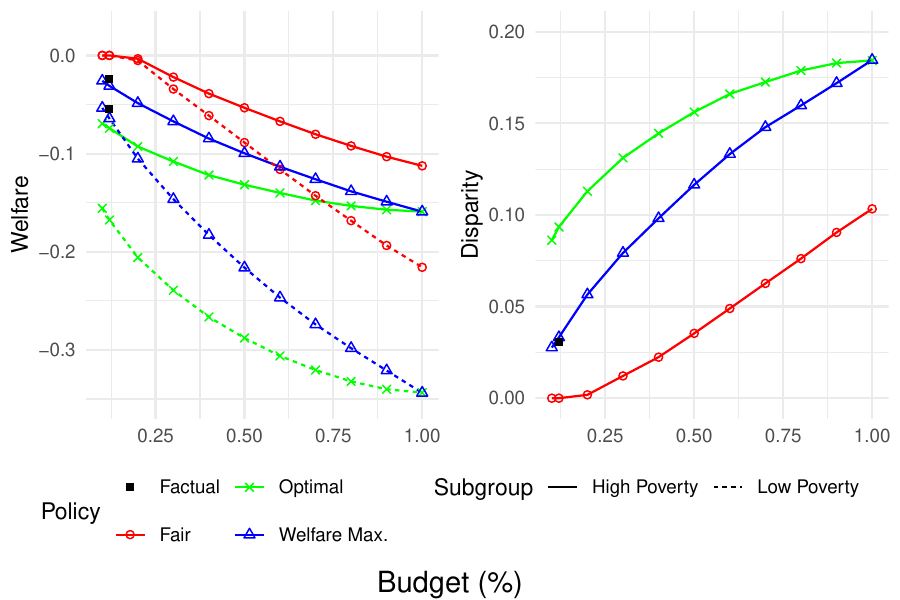}
    \caption{Clean slate policy\label{fig:clean_dual}} 
\end{subfigure}
\hfill
\begin{subfigure}[t]{0.49\textwidth}
    \includegraphics[width=\linewidth]{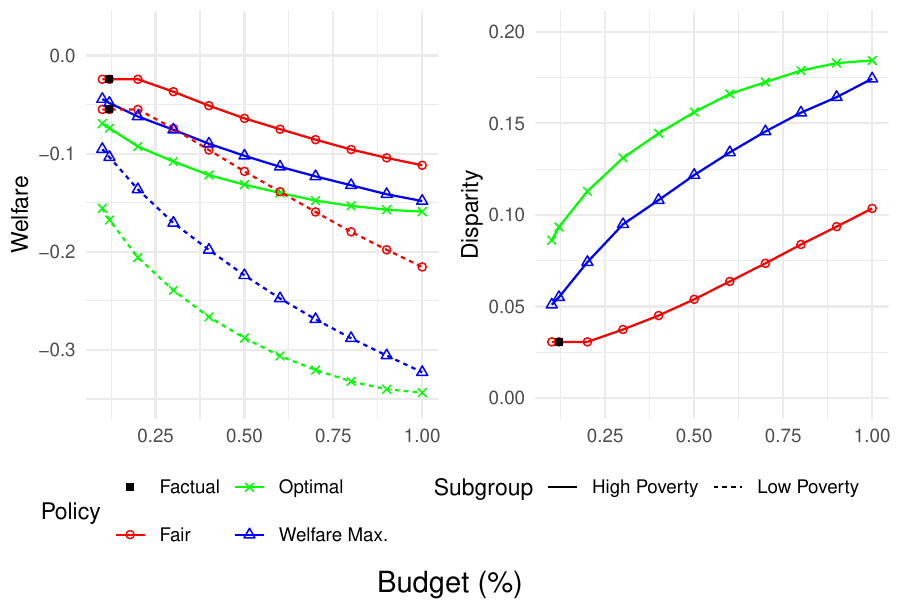}
    \caption{Augmentation policy\label{fig:aug_dual}}
\end{subfigure}    
\caption{Comparison of Clean slate and augmentation policy, across various cost budgets for \dual.}
\label{fig:main_dual}
\end{figure*}

\textbf{Visualization of Optimal and Fair Policy.} 
We investigate the results of the fair policy learner closer with visualizations of the learned policies.
In \cref{fig:panelFair} and \cref{fig:panelOptimal}, we show the fair policy and the optimal policy on the U.S. map, across budgets 10\%-90\%. The plots show that the optimal policy prioritizes certain `focal' plants in the Southeast and Central U.S. (dark circles), while the fair policy will not prioritize any one of these focal plants although the policy still concentrates around this region. This supports our findings from above in which plants that benefit the most \textit{do not} mean the most equitable, and in order to reach fairness, we must treat less.
To highlight the differences and policy concentration, we zoom in on the policies at 50\% budget, in \cref{fig:map_opt_fair_50}.

\begin{figure*}[t!]
\centering
\begin{subfigure}[t]{0.9\textwidth}
    \includegraphics[width=\linewidth]{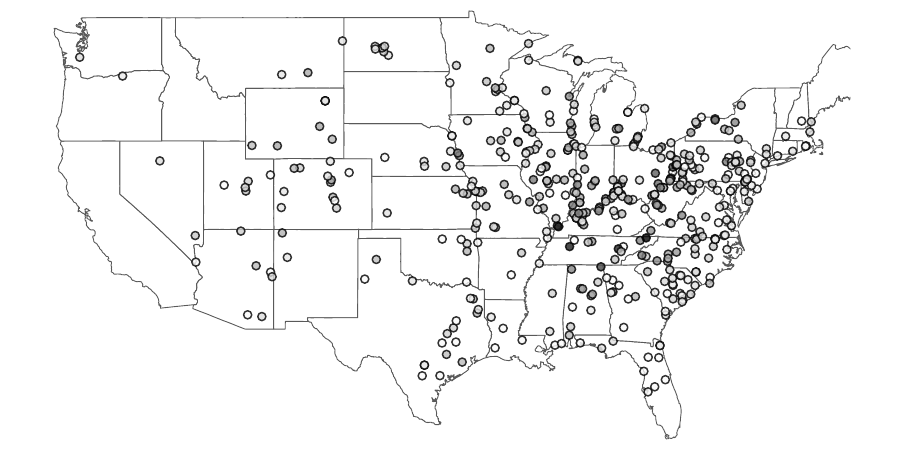}
    \caption{Optimal Policy} 
\end{subfigure}
\hfill
\begin{subfigure}[t]{0.9\textwidth}
    \includegraphics[width=\linewidth]{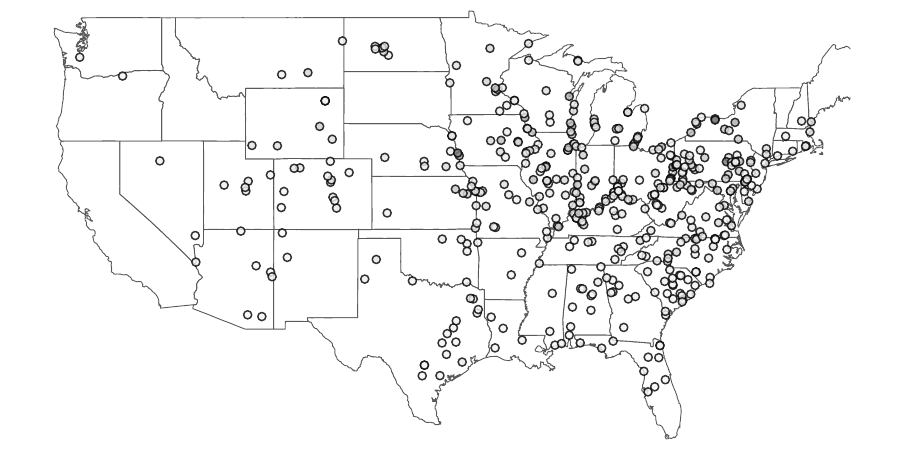}
    \caption{Fair Policy}
\end{subfigure}    
\caption{Comparison of the Optimal and Fair Policy at 50\% budget, for \dual.\label{fig:map_opt_fair_50}}
\end{figure*}

\subsection{Minimum Welfare Constraint}

In our primary analysis, we find that while the fair policy has low \Imb, the welfare benefit to subgroups is lower than what it could be under an optimal or welfare maximization approach. Namely, the benefit to the \dualLow group is decreased under our fair policy learner. 

In this section, we run our fair policy analyses first constraining the minimum welfare for \dualLow to be at least that of the welfare maximization approach, and then we minimize \Imb after this subject to whatever budget constraint. 

We carry out this analysis in \cref{fig:minWelfare_dual} (purple triangles), constraining the \dualLow benefit to be at least as much as that of the Welfare Maximization welfare under 20\% budget (top panel) and 50\% budget (bottom panel).

If we impose the fairness constraint that the minimum welfare received by \dualLow must be at least as high as under the welfare maximization policy (achieved at a 20\% budget), then attaining a fair policy requires increasing the budget to at least 30\%. However, this comes at the cost of a larger induced \Imb compared to the fair policy without this constraint. At higher budget levels, the minimum welfare analysis converges to that of the unconstrained setting.

Similarly, if we target the minimum welfare to be at least as high as the welfare maximization approach at 50\% budget, we require at least 60\% budget to attain a fair policy. This effect is more pronounced than under the analysis above, at the Welfare Maximization performance under 20\% budget; \Imb increases substantially—both relative to the unconstrained fair policy learner and relative to the case where the 20\% budget performance is matched.




\vspace{-1cm}
\begin{figure*}[t!]
\centering
\begin{subfigure}[t]{0.48\textwidth}
    \includegraphics[width=\linewidth]{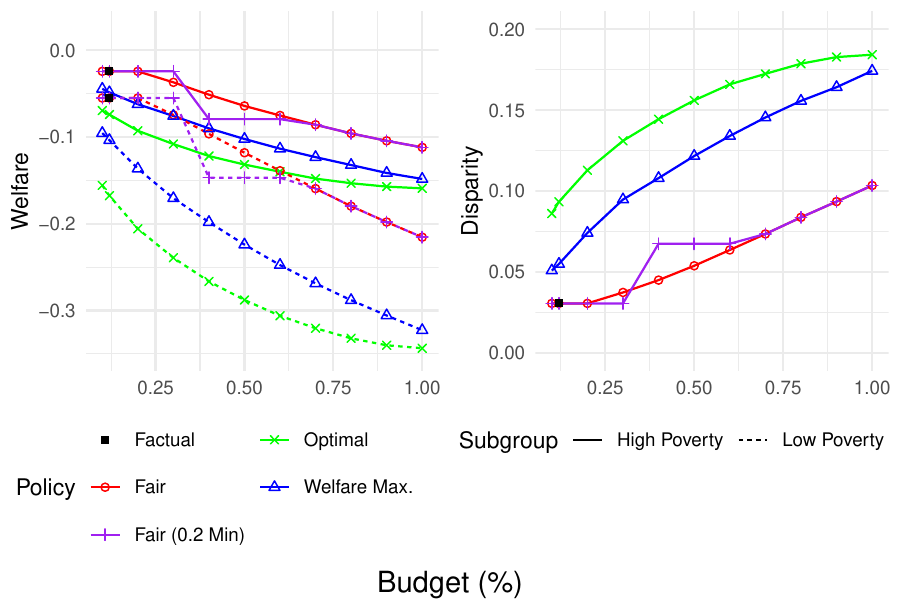}
    \caption{20\% Minimum Constraint\label{fig:minWelfare20_dual}} 
\end{subfigure}
\hfill
\begin{subfigure}[t]{0.48\textwidth}
    \includegraphics[width=\linewidth]{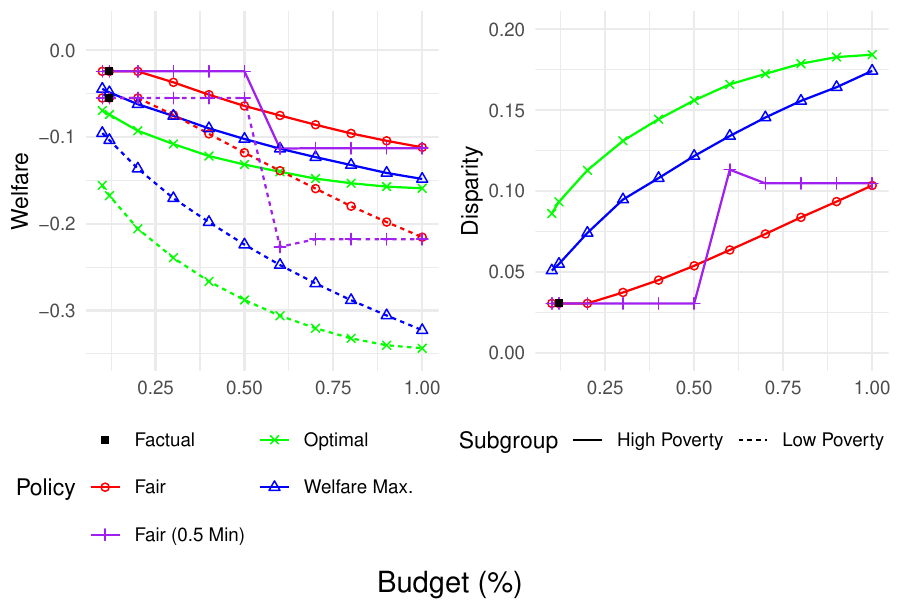}
    \caption{50\% Minimum Constraint\label{fig:minWelfare50_dual}}
\end{subfigure}    
\caption{Comparison of Clean Slate and Augmentation Policy, across various cost budgets for \dual.\label{fig:minWelfare_dual}}
\end{figure*}

\section{Discussion}\label{sec:conclusions}
In this paper, we studied how to assign interventions under BNI to maximize welfare balance while guaranteeing that policies are in line with the principle of ``first, do no harm''. This is of great importance in public health and social science contexts, where interventions are known to disproportionately burden high-risk communities, which can in turn can lead to economic impacts. Balancing such considerations with cost constraints not only benefits the population equitably but also could provide long term economic benefits. 

In contrast to existing interference literature, we design our methods for a setting with BNI spillover effects, without relying on the partial interference assumptions that are often invoked to simplify causal inference in these settings. Our approach also accounts for another unique feature of the BNI context: namely, the fact that treating an intervention unit inherently leads to changes in outcomes on both subgroups, which complicates fair policy considerations compared to settings with conventional data structures. We prove regret bounds, validate our theoretical results in simulation studies, and apply our methods to study the tradeoffs in learning optimal and fair policies for the installation of emissions reducing technologies on coal-fired power plants. 

Fundamentally, our analyses have supported that \textit{fairness comes at a price}. We must pay higher cost or reduce our overall welfare in order to achieve fairness. The contrast between optimality and fairness highlights the many considerations policymakers must balance when designing environmental policies.

While in this work we focus on environmental policies, our methods applicability is not limited to this endeavor. In fact, the proposed method can be extensively applied to any real-world data setting with a similar BNI data structure in order to answer general questions on fair, cost-effective policy design. In healthcare, such BNI structures naturally arise: hospital readmission programs involve spillovers between facilities and patient communities \citep{zuckerman2016readmissions, baicker2013oregon}; mental health interventions link treatment centers with catchment areas creating possible community spillovers \citep{baicker2013oregon}; opioid treatment programs connect facilities with counties, generating spillovers through drug supply and peer networks \citep{pacula2018supply}. Beyond healthcare, in marketplace settings, one may model the causal effects of online advertising campaigns from businesses on consumers. Here, businesses and consumers are distinct node types and high spillover may occur depending on the advertising strategy \citep{doudchenko2020causal}. Across these domains, policymakers must navigate trade-offs between efficiency and equity while accounting for BNI which is precisely the challenge our framework addresses.

Our work has several limitations that could be addressed in future work. First, we rely on some strong assumptions about the outcome model form, and we assume that intervention costs are fixed and known. Future work could consider accommodating unknown treatment costs and general model forms. Second, we have not conducted inference relating to the fair policies. To increase utility to policymakers, considering statistical uncertainty (inference) along the frontier and a greater understanding of optimality-fairness gaps is also an interesting future direction \citep{LiuMolinari24,auerbach2025testingfairnessaccuracyimprovabilityalgorithms}.

\textbf{Data Availability. }
Patient health information is protected as specified in the Data Use Agreement between the Center for Medicare \& Medicaid Services (CMS) and our institution, which additionally determined: (1) waiver of HIPAA authorization, (2) waiver/alteration of the consent process, since the use of Medicare data does not require obtaining consent from individual Medicare beneficiaries. Due to the sensitive nature of this data, it is not available for public use. The power plant data and ZIP code level confounding data are publicly available as noted in the text, and simulation code will be made available at the following Github repository \href{https://github.com/NSAPH-Projects/emissions-ihd-bni_pareto}{\texttt{https://github.com/NSAPH-Projects/emissions-ihd-bni\_pareto}}.

{ \singlespacing
\bibliography{refs}
}

\appendix


\renewcommand\Authfont{\fontsize{12}{14.4}\selectfont} 
\renewcommand\Affilfont{\fontsize{11}{10.8}\itshape} 

\crefname{section}{Appendix}{Appendices}

\counterwithin{equation}{section}
\counterwithin{figure}{section}
\counterwithin{table}{section}
\pagenumbering{arabic}
\setcounter{page}{1}

\renewcommand{\thesection}{Appendix~\Alph{section}}
\appendix 

\if1\anon
{
\begin{center}
    \textbf{\large Supplementary Materials to\\ ``Fair Policy Learning under Bipartite Network Interference: Learning Fair and Cost-Effective Environmental Policies''}\\ \vspace{0.25cm}
    \normalsize Raphael C. Kim$^1$, Rachel C. Nethery$^1$, Kevin L. Chen$^1$, and Falco J. Bargagli-Stoffi$^{2,1}$ \\
    {\small $^1$Department of Biostatistics, Harvard T.H. Chan School of Public Health, 677 Huntington Ave, 02215, MA, USA}\\
    {\small $^2$Department of Biostatistics, UCLA Fielding School of Public Health, 650 Charles E Young Dr S, 90095, CA, USA}
\end{center}
} \fi

\if0\anon
{
  \bigskip
  \bigskip
  \bigskip
  \begin{center}
    \textbf{\large Supplementary Materials to\\ ``Fair Policy Learning under Bipartite Network Interference: Learning Fair and Cost-Effective Environmental Policies''}\\ \vspace{0.25cm}
\end{center}
  \medskip
} \fi


\section{Proofs} \label{sec:proofs}

\subsection{Spatial Setup and Mixing Assumptions}\label{sec:SpatialSetup}

    We rely on mixing assumptions to conclude regret bounds for our policy.
    We formalize our spatial setup as in \cite{Jenish2009}. Let $D_n \subseteq D \subseteq \mathbb{R}^2$ for $D_n$ a finite subset. $\Xouti$ lie in $D_n$. Define $L(i)$ as the location function; $L(i)$ returns the 2d coordinate location of unit $i$ in $D$.
    Take metric $\rho(i, i') =\max_{1 \leq k \leq 2} |L(i)_k - L(i')_k|$, which induces a norm $|i| = \max_{1 \leq k \leq 2} |L(i)_k|$, for $L(i)_k$ denoting the $k$-th component of the coordinate location of unit $i$. The distance between two subsets $U,V \subset D$ is defined as $\rho(U,V) = \inf \{\rho(i, i') : i \in U, i' \in V \}$, and the cardinality of a finite subset $U$ of $D$ is denoted $|U|$. Further define the spatial mixing coefficient as follows (overloading notation on $\beta$):

        For $U \subset D_n, V \subset D_n$, define $\sigma_n(U)=\sigma(\Xouti: i \in U)$, the sigma field over $U$, and $\beta_n(U,V)=\beta(\sigma_n(U), \sigma_n(V))=\frac{1}{2} \sup \sum_{U'_s} \sum_{V'_{r}} |\pr(U'_s \cap V'_r)-\pr(U'_s)\pr(V'_r)|$ where $\{ U'_s \}, \{ V'_r\}$ are finite partitions of $U$ and $V$ respectively. Then, the $\beta$ mixing coefficient on $D_n$ is given by:
        $$\beta_{s,s'}(r)=\sup\{ \beta_n(U,V): |U| \leq s, |V| \leq s', \rho(U,V) \geq r \}$$

        Also define $\overline{\beta}_{s,s'}(r)=\sup_n \beta_{s,s',n}(r)$ to account for sampling region variability.

    We assume the following.\\
    \textbf{Spatial Assumptions}
    \begin{enumerate}[label=\bfseries (Sp\arabic*)]
        \item \label{ass:samplingRegime} \textit{Infinite Sampling Regime}: \label{ass:infiniteSampling} $D$ is infinitely countable and all elements are at least $\rho_0 \geq 1$ apart for all $i, i' \in D$.
        \item \label{ass:mixing} \textit{Spatial Mixing}: For $\tau > 2$,
        $r \overline{\beta}_{s,s'}(r) \leq (s+s') \cdot r^{-\tau}$
    \end{enumerate}    

    These assumptions are standard in spatial statistics literature \citep{Jenish2009}, roughly saying that we continue to observe outcome units growing in space, and the farther two units are from one another, the weaker the dependence is. We will need the technique of blocking, first introduced in \cite{Bernstein1927}. The idea is to partition our space into blocks and consider `every other block' to induce rough independence between blocks with some space apart. To that end, we define the partition as follows.

    Divide $D_n$ into blocks with area $d_n^2$ yielding $L_n$ blocks. For $z \in \{ 0, 1 \}^2$, let $type(i)$ be whether or not sample $i$ lies in some block $l$ corresponding to `type z', or spatial locations with block $l$ having the 1st coordinate being even/odd (0/1) and the second coordinate being even/odd (0/1). This permits us to employ the independent blocking from Bernstein on $\mathcal{B}_{z,l}=\{ \Xouti \in \mathcal{B}_{l} \mbox{ and }type(i)=z, \forall l  \}$, or the data samples that are of type $z$ and lie in $\mathcal{B}_l$. The details of how to use this construction are outlined in the results below.
    
\subsection{Proof of \texorpdfstring{\cref{thm:paretoFrontNeg}}{Proposition}}
\begin{proof}\label{thm:ParetoFront}
    The proof follows in a similar manner to Lemma 2.1 of \cite{viviano2024fair}, which deals with a maximization problem instead of minimization. The proof for the minimization follows similarly. Define $$ \tilde{\Pi}_{0}= \{ \pi : \arg \sup_{\pi \in \Pi} \sum_{s \in [|S|]} \Wweight_s W_{s}(\pi), \pmb{\Wweight} \in \mathbb{R}^{|S|}_+  \}  $$

    We wish to show that $\tilde{\Pi}_0=\Pi_0$.
    \noindent $ \tilde{\Pi}_{0} \subseteq \Pi_{0} $: This follows as $\tilde{\Pi}_{0}$ consists of all policies that minimize the welfare for each subgroup value $s \in S$. Deviations from a policy on this set would not improve a particular subgroup without harming another by definition.
    \\\\
    $ \Pi_{0} \subseteq \tilde{\Pi}_{0}  $: Let $F= \{ \pmb{x} \in \mathbb{R}^{|S|}_+ : \exists \pi \in \Pi, x_s \leq W_{s}(\pi) \}$. $F$ is convex because $W_{s}$ is concave by linearity in $\pi$, and $F$ is the hypograph of $W_{s}$. Further $F$ is non-empty since $\mathbf{0} \in F$.
    Define $G= \{ d_s + W_{s}(\pi_{0}): \pmb{d} \in \mathbb{R}^{|S|}_+, \pi_{0} \in \Pi_{0} \}$, the welfare values that strictly dominate $\pi_{0}$. $G$ is also non-empty and convex. 

    By definition of a policy $\pi_{0}$ that is Pareto-optimal, $F \cap G = \emptyset$. We can then apply the separating hyperplane theorem to conclude that 
    $\exists \pmb{\Wweight} \in \mathbb{R}^{|S|}$ such that for $x \in F, d \in \mathbb{R}^{|S|}_+$, $\sum_{s \in [|S|]} \Wweight_s x_s \leq \sum_{s \in [|S|]} \Wweight_s W_{s}(\pi_0)+d_s$. 
    Normalizing by $\pmb{\Wweight}$, this implies that for $\pmb{\Wweight} \in \Delta^{|S|}$, $\sum_{s \in [|S|]} \Wweight_s W_{s}(\pi) \leq \sum_{s \in [|S|]} \Wweight_s W_{s}(\pi_{0})$, for any $\pi \in \Pi$ yielding the result. 
\end{proof}

\subsection{Helper Lemmas}
In this section, we establish some helper lemmas on concentration bounds.

\begin{lemma}{Convergence of the Summary Functional.}\label{lemma:etaCv}

    Assume boundedness \labelcref{ass:bounded}, spatial mixing \labelcref{ass:infiniteSampling}-\labelcref{ass:mixing}, and $\eta$ satisfies \labelcref{ass:summaryFxal}. Then, for some deterministic $\eta_j \in \R^{p+1}$,
    $$ \etaMap \xrightarrow{p} \eta_j$$
\end{lemma}
\begin{proof}
    This follows directly from Theorem 3 of \cite{Jenish2009}.
    
\end{proof}

\subsection{Proof of \texorpdfstring{\cref{lemma:welfarePopEta}}{Lemma}}\label{pf:welfarePopEta}

\begin{proof}
    Let $\eta_j$ be the deterministic limit of $\etaMap$ as in \cref{lemma:etaCv}. Then, 
    \begin{align*}
        \hat{W}_s(\pi) &= \frac{1}{n}\sum_{i=1}^n \frac{\I{S_i=s}}{p_s} f_A(\Xouti, \outCoeftx_0) \frac{1}{J} \sum_{j=1}^J \Tinst \pi(A_j=1\mid \etaMap, \Xintj) \\
        &= \frac{1}{n}\sum_{i=1}^n \frac{\I{S_i=s}}{p_s} f_A(\Xouti, \outCoeftx_0) \frac{1}{J} \sum_{j=1}^J \Tinst [\pi(A_j=1\mid \etaMap, \Xintj) -  \pi(A_j=1\mid \eta_j, \Xintj)] \\
        & \quad\quad\quad + \frac{\I{S_i=s}}{p_s} f_A(\Xouti, \outCoeftx_0) \frac{1}{J} \sum_{j=1}^J \Tinst \pi(A_j=1\mid \eta_j, \Xintj)
    \end{align*}

    By the Donsker Policy class \labelcref{ass:donskerPolicy} and boundedness of our interference map and outcome \labelcref{ass:bounded}, we have
    \begin{align*}
        \sup_{\pi} |\frac{1}{n}\sum_{i=1}^n & \frac{\I{S_i=s}}{p_s} f_A(\Xouti, \outCoeftx_0) \frac{1}{J} \sum_{j=1}^J \Tinst [\pi(A_j=1\mid \etaMap, \Xintj) -  \pi(A_j=1\mid \eta_j, \Xintj)]| \\
        & \leq |\frac{1}{J} \sum_{j=1}^J M^2 [\pi(A_j=1\mid \etaMap, \Xintj) -  \pi(A_j=1\mid \eta_j, \Xintj)]|\\
        & = \bigO_P(n^{-1/2})
    \end{align*} 

    Hence, the result follows.
\end{proof}

\subsection{Proof of \texorpdfstring{\cref{lemma:TE_Est_Concentration}}{Lemma}}\label{pf:TE_Est_Concentration}

\begin{proof}
    We show this using the blocking technique described in \cref{sec:SpatialSetup}. We then bound each block rademacher complexity in terms of the spatial region, and conclude rates accordingly. This follows in a similar manner to \cite{kim2024optimalenvironmentalpoliciespolicy}.
    \begin{enumerate}
        \item  \textbf{Deviation Bound on Independent Blocks via Spatial $\beta$-Mixing rate.}
            Let $\B_{z,l}'$ be independent copies of $\B_{z,l}$, and $\B_z$ be the blocks $\B_{z',l}$ with $z=z'$ (and $\B_z'$ similarly).
            By Proposition 1 of \cite{KuznetsovMohri17}, we have for an $M-$bounded function $h$ and fixed $z$,
            \begin{align*}
                \E[h(\B_{z}')-h(\B_{z})] \leq L_n \cdot |D_n| \sup_{q,q'}\overline{\beta}_{q,q'}(d_n) 
            \end{align*}
             Take $h$ to be the block averages for block type $z$, or $h(\B_{z})=\E[\frac{1}{|\B_{z}|}\sum_{l} \sum_{i \in \B_{z,l}} \Tinst (f_A(\Xouti, \outCoeftx) -\E  [  f_A(\Xouti, \outCoeftx))]]$. Then it follows from the union bound that
            \begin{align*}
                \pr[\frac{1}{n} \sum_{i =1}^n \Tinst &  (f_A(\Xouti, \outCoeftx)  - \E [f_A(\Xouti, \outCoeftx)] > \delta] \\
                & \leq \sum_z \pr[h(\B_z')-\E h(\B_z') > \delta - \E h(\B_z')] + L_n \cdot \sup_{s,s'}\overline{\beta}_{s,s'}(d_n) 
            \end{align*}
        \item \textbf{Concentration Bound using Mcdiarmid.}
            Note that a one \textit{block} difference yields an empirical deviation of $M \frac{d_n}{n}$. By Mcdiarmid's inequality, 
            \begin{align*}
                \pr[\frac{1}{|\B_z'|} \sum_{i \in \B_z'} \Tinst (f_A(\Xouti, \outCoeftx) - \E [f_A(\Xouti, \outCoeftx)] > \delta] & \leq \exp{(\frac{-2 |D_n|^2 (\delta+\E [f_A(\Xouti, \outCoeftx)])^2}{L_n d_n^2 M^2})}
            \end{align*} 
        \item   \textbf{Empirical Process Bound.}
            It remains to bound $\E h(\B_z')$. By symmetrization and Dudley's integral, invoking \labelcref{ass:donskerOutcome}
            \begin{align*}
                \E h(\B_z') &\leq \E \sup_{\outCoeftx} \frac{1}{|\B_z'|} \sum_{l} \sum_{i \in B_{z,l}' }\sigma_i f_A(\Xouti, \outCoeftx) \\
                & \lesssim \frac{M}{\sqrt{n}} 
            \end{align*}
            for $\sigma_i$ i.i.d. Rachemacher r.v.'s.
        
            Integrating out the tail and employing \labelcref{ass:mixing}, we have :
            $$ \E \sup_{\outCoeftx} |\frac{1}{n} \sum_{i=1}^n \Tinst f_A(\Xouti, \outCoeftx_0) - \E[\Tinst f_A(\Xouti, \outCoeftx_0)]| \lesssim \sqrt{\frac{M}{|D_n|}} + d_n \sqrt{\frac{M}{|D_n|}}$$
        
            We can choose $d_n$ such that $\frac{d_n}{|D_n|}=o_p(1)$, or $d_n$ grows slower than $D_n$. Noting that WLOG, we can take $\rho_0 = 1$, the result follows.
    \end{enumerate}
\end{proof}

\subsection{Proof of \texorpdfstring{\cref{thm:ParetoFrontEst}}{Theorem}}
\begin{proof}\label{pf:ParetoFrontEst}
    Recall, we take $K = \sqrt{n}$, meaning $\frac{\lambda}{K}=\frac{\lambda}{\sqrt{n}}$. By the triangle inequality,
    \begin{align*}
        \E[\sup_{\pmb{\Wweight}, \pi} |\sum_s \Wweight_s W_s(\pi) - \inf_{a_k} \{ \sum_s \Wweight_{ks} \hat{W}_s(\pi) \}|] &\lesssim \underbrace{\E \sup_{\pmb{\Wweight}, \pi} |\sum_s \Wweight_s W_s(\pi) - \Wweight_s \hat{W}_s(\pi)|}_{(I)} \\
        &\quad\quad\quad + \frac{\lambda}{\sqrt{n}} \\
        &\quad\quad\quad + \underbrace{\E \sup_{\pmb{\Wweight}, \pi}|\sum_s \Wweight_s \hat{W}_s(\pi) - \inf_{\pmb{\Wweight}_k} \sum_s \Wweight_{ks} \hat{W}_s(\pi)|}_{(II)} 
    \end{align*}

    \textbf{Term $I$ } We bound this using a concentration result based on \labelcref{ass:samplingRegime}-\labelcref{ass:mixing}. We decompose term $I$, and bound each of the terms.
    \begin{align*}
        \E \sup_{\pmb{\Wweight},\pi} |\sum_s \Wweight_s W_s(\pi)-\Wweight_s \hat{W}_s(\pi)| &\leq \sum_s \underbrace{\E \sup_{\pmb{\Wweight},\pi} |\Wweight_s W_s(\pi)-\Wweight_s W_{s,n}(\pi)|}_{(A)} + \underbrace{\E \sup_{\pmb{\Wweight},\pi} |\Wweight_s W_{s,n}(\pi) - \Wweight_s \hat{W}_s(\pi)|}_{(B)}
    \end{align*}

    \textbf{Term A}\\ 
    By \cref{lemma:welfarePopEta}, it suffices to consider the following, where recall we have let $\EstTotEff_j(s) = \frac{1}{n} \sum_{i=1}^n \I{S_i=s} \Tinst f_A(\Xouti, \outCoeftx_0) $.
    \begin{align*}
        A &= |\frac{1}{J} \sum_{j=1}^J \frac{1}{n}\sum_{i=1}^n \frac{\I{S_i=s}}{p_s}\Tinst f_A(\Xouti, \outCoeftx_0) \pi(A_j=1\mid \eta_j, \Xintj)  \\
        & \quad\quad\quad\quad -\E[\Tinst \frac{\I{S_i=s}}{p_s} f_A(\Xouti, \outCoeftx_0) \pi(A_j=1\mid \eta_j, \Xintj)]| \\
        &= |\frac{1}{J} \sum_{j=1}^J \EstTotEff_j(s) \pi(A_j=1\mid \eta_j, \Xintj) -\E[\EstTotEff_j(s)  \pi(A_j=1\mid \eta_j, \Xintj)]| \\
        &= |\frac{1}{J} \sum_{j=1}^J \EstTotEff_j(s)  \pi(A_j=1\mid \eta_j, \Xintj) - \E[\EstTotEff_j(s)]\pi(A_j=1\mid \eta_j, \Xintj)|\\
        & \quad\quad\quad\quad + |\E[\EstTotEff_j(s) ]\pi(A_j=1\mid \eta_j, \Xintj) -\E[\EstTotEff_j(s) \pi(A_j=1\mid \eta_j, \Xintj)]| \\
        &= |\frac{1}{J} \sum_{j=1}^J \pi(A_j=1\mid \eta_j, \Xintj) [\EstTotEff_j(s)  - \E[\EstTotEff_j(s) ]]|\\
        & \quad\quad\quad\quad + |\E[\EstTotEff_j(s) ](\pi(A_j=1\mid \eta_j, \Xintj) -\E[\pi(A_j=1\mid \eta_j, \Xintj)])| 
    \end{align*}
    
    Now, note that by \cref{lemma:TE_Est_Concentration} and boundedness of $\frac{\I{S_i=s}}{p_s}$, 
    $$ |\EstTotEff_j(s)  - \E[\EstTotEff_j(s) ]| \lesssim \frac{M}{\sqrt{n}} $$

    Further, by \labelcref{ass:donskerPolicy},
    $$|\frac{1}{J}\sum_{j=1}^J (\pi(A_j=1\mid \eta_j, \Xintj) -\E[\pi(A_j=1\mid \eta_j, \Xintj)]| \lesssim \frac{1}{\sqrt{n}}$$ 

    Hence, $$ |A| \lesssim \frac{M}{\sqrt{n}}$$
    
    \textbf{Term B}
    \begin{align*}
        \E \sup_{\pmb{\Wweight},\pi} |\Wweight_s W_{s,n}(\pi) - \Wweight_s \hat{W}_s(\pi)| &\leq \E \sup_{\pmb{\Wweight},\pi} |\Wweight_s W_{s,n}(\pi)- \Wweight_s \E[\hat{W}_s(\pi)]| + |\Wweight_s \E[\hat{W}_s(\pi)] - \Wweight_s \hat{W}_s(\pi)| \\
        & \leq \E \sup_{\pmb{\Wweight},\pi} |\Wweight_s W_{s,n}(\pi)- \Wweight_s \E[\hat{W}_s(\pi)]| + \frac{M}{\sqrt{n}}
    \end{align*}
    where the second line follows from \cref{lemma:TE_Est_Concentration}. We now bound the remaining term. 
    \begin{align*}
        \E \sup_{\pmb{\Wweight},\pi} |\Wweight_s W_{s,n}(\pi)- \Wweight_s \E[\hat{W}_s(\pi)]| &= \E \sup_{\pmb{\Wweight},\pi} |\frac{1}{n} \sum_{i=1}^n \frac{1}{J} \sum_{j=1}^J \Tinst \frac{\I{S_i=s}}{p_s} f_A(\Xouti, \outCoeftx_0)  \pi_j \\
        &\quad\quad\quad\quad -\E[\Tinst \frac{\I{S_i=s}}{p_s} f_A(\Xouti, \hat{\outCoeftx}_0)  \pi_j] |
    \end{align*}
    The structure is the same as term A, except $\outCoeftx$ is estimated. Since \cref{lemma:TE_Est_Concentration} is a supremum bound over $\outCoeftx$, $$ |B| \lesssim \frac{M}{\sqrt{n}}$$
    
    \textbf{Term $II$} Note that $II$ represents the maximum difference between the approximated grid at resolution $K$ and the true minimizer, under welfare $\hat{W}_s(\pi)$. Thus, we can bound as follows, using \labelcref{ass:bounded}:
    $$ II \lesssim \frac{1}{K} \sup_{\pi,s}|\hat{W}_s(\pi)| = \frac{M}{K}$$

    \textbf{Final Result.}
    Together, we have 
    \begin{align*}
        \sup_{\Wweight, \pi}|\sum_s \Wweight_s W_s(\pi) - \inf_{\Wweight_k} \{ \sum_s \Wweight_{ks} \hat{W}_s(\pi) \}| &\lesssim \frac{\lambda + M}{\sqrt{n}}
    \end{align*}
\end{proof}

\subsection{Proof of \texorpdfstring{\cref{thm:ParetoFrontSupp}}{Theorem}}

\begin{proof}\label{pf:ParetoFrontSupp}
    Let $\overline{W}_{\Wweight}=\inf_\pi \Wweight W_0(\pi)+(1-\Wweight) W_1(\pi)$, and $\overline{W}_{k,n}=\inf_\pi \Wweight_k W_{0,n}(\pi) +(1-\Wweight_k) W_{1,n}(\pi)$

    Since $k\in [K]$ are equally spaced, it suffices to show that the probability of deviation between any object on the front and the true empirical welfare, up to the slack $\lambda$, is bounded by $\gamma$:
    $$ \mathbb{P}(\forall \Wweight \in (0,1), \max_{k \in [K]}| \overline{W}_\Wweight  - \overline{W}_{k,n}| \geq \frac{\lambda}{\sqrt{n}} )  \leq \gamma$$
    Note that for any $\pmb{\Wweight}, k$ (by equal-spacing of the grid),
    \begin{align*}
        |\overline{W}_\Wweight  - \overline{W}_{k,n} |&= 
        |\overline{W}_\Wweight - \sum_s \Wweight_s \hat{W}_s(\pi) + \sum_s \Wweight_s \hat{W}_s(\pi) - \overline{W}_{k,n} | \\
        &\leq |\overline{W}_\Wweight - \sum_s \Wweight_s \hat{W}_s(\pi)| + |\sum_s \Wweight_s \hat{W}_s(\pi) - \overline{W}_{k,n}| 
    \end{align*}
    
    Then by Markov's inequality,
    \begin{align*}
        \mathbb{P}(\forall \Wweight \in (0,1), \max_{k \in [K]}, |\overline{W}_\Wweight  - \overline{W}_{k,n} | \geq \frac{\lambda}{\sqrt{n}}  )  &\leq \frac{ \lambda \E[|\overline{W}_\Wweight  - \overline{W}_{k,n} |]}{\sqrt{n}}  \\
        & \leq C M \lambda 
    \end{align*}
    The last step follows by \cref{thm:ParetoFrontEst} for some universal constant $C$. It suffices to take $\gamma \geq C M \lambda$ or $\lambda \leq \frac{\gamma}{C M} $.
\end{proof}

\subsection{Proof of \texorpdfstring{\cref{thm:RegBd}}{Theorem}}

\begin{proof}\label{pf:RegBd}
    Let $\Imb(\pi)=|W_1(\pi)-W_0(\pi)| $ and $\widehat{\Imb}(\pi)=|\hat{W}_1(\pi)-\hat{W}_0(\pi)|$.

    The minimizer of $\Imb(\pi)$ is $\inf_{\pi \in \Pi_0} \Imb(\pi)$. Denote the solution to the objective (\cref{eq:optProb}) by $\hat{\pi}_{\lambda}$. 

    By \cref{thm:ParetoFrontSupp}, for $\gamma(n)$ chosen as $\frac{1}{\sqrt{n}}$,
    \begin{align*}
        \E[| \Imb(\hat{\pi}_\lambda) - & \inf_{\pi \in \Pi_0}\Imb(\pi) |] = \pr(\{ \Pi_0 \subseteq \hat{\Pi}_0 \})\E[|\Imb(\hat{\pi}_\lambda) -\inf_{\pi \in \hat{\Pi}_0(\lambda)} \Imb(\pi)|] \\
        & \quad\quad\quad + (1-\pr(\{ \Pi_0 \subseteq \hat{\Pi}_0 \})) \E[|\Imb(\hat{\pi}_\lambda) -\inf_{\pi \not \in \hat{\Pi}_0(\lambda)} \Imb(\pi)|]  \\
        & = (1-\gamma(n)) \E[|\Imb(\hat{\pi}_\lambda) -\inf_{\pi \in \hat{\Pi}_0(\lambda)} \Imb(\pi)|] \\
        & \quad\quad\quad + \gamma(n) \E[|\Imb(\hat{\pi}_\lambda) -\inf_{\pi \not \in \hat{\Pi}_0(\lambda)} \Imb(\pi)|]  \\
        & \lesssim (1-\frac{1}{\sqrt{n}}) \E[|\Imb(\hat{\pi}_\lambda) -\inf_{\pi \in \hat{\Pi}_0(\lambda)} \Imb(\pi)|] \\
        & \quad\quad\quad + \frac{1}{\sqrt{n}} M  
    \end{align*}

    It suffices to bound  
    \begin{align*}
        \E[|\Imb(\hat{\pi}_\lambda) -\inf_{\pi \in \hat{\Pi}_0(\lambda)} \Imb(\pi)|]  & = \E[|\Imb(\hat{\pi}_\lambda) - \widehat{\Imb}(\hat{\pi}_\lambda)|] \\
        &\quad\quad + \E[|\widehat{\Imb}(\hat{\pi}_\lambda)  - \inf_{\pi \in \hat{\Pi}_0(\lambda)} \Imb(\pi) |] \\
        & \leq 2 \E \sup_{\pi \in \hat{\Pi}_0(\lambda)} |\Imb(\pi)-\widehat{\Imb}(\pi)| \\
        & \leq 2 \E \sup_{\pi \in \Pi} |\Imb(\pi)-\widehat{\Imb}(\pi)| \\
        & \lesssim \frac{M+\lambda}{\sqrt{n}}
    \end{align*}
    
    where the last step follows by \cref{thm:ParetoFrontEst}.
\end{proof}
 
\section{Additional Simulation Details}\label{sec:sim_details}

In our simulation, parameters are chosen to align closely with our real world application. $\intCoef_0$ is chosen so that simulated average of propensities $\frac{1}{J} \sum_{j=1}^J e_j$ is within 0.01 of the observed empirical average of treatments in the dataset (see final row of \cref{tab:pp_covs}). Further, $\outCoef_0=(\outCoefb,\outCoeftx)$ is chosen so that the simulated average $\frac{1}{n}\sum_{i=1}^n Y_i$, based on the expected exposure level $\bar{A}=\frac{1}{J}\sum_{j=1}^J \Tinst e_j$, is within 0.01 of the empirical average in the dataset (see \cref{tab:zip_covs}). 

\textbf{Outcome Model.}
We simulate $f_0, f_A$ as linear in $\Xouti$, where $\alpha_{\mathrm{intercept}} \in \mathbb{R}$ and $\outCoefb_1 \in \mathbb{R}^{\mathrm{dim}(\Xouti)}$. Similarly, $\beta_{\mathrm{intercept}} \in \mathbb{R}$ and $\outCoeftx_1 \in \mathbb{R}^{\mathrm{dim}(\Xouti)}$. 
\begin{equation} \label{eq:outBaseModelSim}
    f_0(\Xouti, \alpha)=\alpha_{\mathrm{intercept}} + \Xouti  \outCoefb_{1}
\end{equation}
\begin{equation}\label{eq:outTxModelSim}
    f_A(\Xouti, \beta)= \beta_{\mathrm{intercept}} +  \Xouti \outCoeftx_{1} 
\end{equation}

\textbf{Propensity Score Model.}
 We model $e_j$ using a logistic regression model where $\gamma_{\mathrm{intercept}} \in \mathbb{R}$ and $\intCoef_1 \in \mathbb{R}^{\mathrm{dim}(\Xouti)}$.
\begin{align*} \label{eq:intPropModelSim}
    \log \left(\frac{e_j}{1-e_j}\right) &= \gamma_{\mathrm{intercept}} + \Xintj \intCoef_{1}
    \numberthis
\end{align*}

Our simulation setup is as follows:
\begin{enumerate}
    \item \textbf{ Preprocess data}: Standardize the covariates of $\Xouti, \Xintj$.
    \item \label{sim:stepGenTx} \textbf{ Generate intervention unit treatments}: $A_j \sim \mathrm{Bernoulli}(e_j)$, where $e_j$ follows \cref{eq:intPropModelSim}.  
    \item \textbf{ Compute exposure mapping}: $\Ei=\frac{1}{J} \sum_{j \in [J]}\Tinst A_j$
    \item \textbf{ Generate outcomes}: $Y_i(\Ei)$ such that the SNR is \SNR, i.e. 
        $$Y_i=f_0(\Xouti, \outCoefb) + \Ei \cdot f_A(\Xouti,\outCoeftx) + \epsilon_i$$ 
        where 
         $\mathbb{V}(\epsilon_i)=\mathbb{V}\left(\frac{\E[Y_i \mid \Xouti, \T_i, \Ei; \outCoef]}{\SNR}\right)$ 
        and $f_0$, $f_A$ follow \cref{eq:outBaseModelSim,eq:outTxModelSim} respectively. 
    \item \textbf{Estimate parameters using \al}: $\outCoef=(\outCoefb, \outCoeftx)$
    \item \textbf{Run Fair Policy Learning.}\label{sim:fairPolicy}
    \begin{itemize}
        \item Fair Policy Learning, as described in \cref{sec:method}.
        \item Welfare Maximization, from the Utilitarian Perspective as described in \cite{KitagawaTetenov18}
    \end{itemize}
    \item \textbf{ Iterate}: Repeat steps (\labelcref{sim:stepGenTx}-\labelcref{sim:fairPolicy}) 1000 times.
    \begin{enumerate}
        \item Compute $W_0(\pi), W_1(\pi), \Imb(\pi)$ for each fair policy learner.
        \item Average these metrics across 1000 simulations.
    \end{enumerate}
\end{enumerate}

The specific simulation parameters we use are as follows:
\begin{align*}
    \pmb{\theta}_0=( &0.649,0.963,0.33,0.411,-0.481,0.733,0.566,0.343, \\
    &0.058,-0.934,-0.277,-0.995,0.709,0.419,-0.505,\\
    &0.517,0.03,-0.723,0.854,-0.496,-0.393,0.316,\\
    &0.487,-0.444,-0.653,-0.052,0.931,0.143)
    \\
    \pmb{\gamma}_0=(&-0.997,-0.447,-0.04,0.021,0.806,-0.689,-0.823, \\
    &-0.909,-0.658,-0.101,0.908,0.911,0.193,0.408,\\
    &-0.835,0.392,0.625,0.13,0.022,0.073)
\end{align*}

\section{Descriptive Statistics}
We compute basic descriptive statistics on our outcome and intervention dataset.

\setcounter{table}{0}
\begin{table}[H]
\begin{tabular}{lll}
Variable                             & Mean   & Range          \\ \hline
Total NO$_x$ controls               & 2.83  & (0, 24)        \\
log(Heat input)                     & 14.52 & (8.98, 17.32) \\
log(Operating time)                  & 7.19  & (5.39, 8.93)   \\
\% Operating capacity                 & 0.66  & (0.072, 1.17)   \\
\% Selective non-catalytic reduction  & 0.23  & [0, 3]         \\
ARP Phase II                        & 0.70  & \{0, 1\}  \\ \hline
Scrubbed                        & 0.23  & \{0, 1\}  \\ \hline
\end{tabular}
\caption{\textmd{Summary of intervention level covariates from power plant data $(J=459)$.}}\label{tab:pp_covs}
\end{table}

\setlength{\floatsep}{2pt}
\setlength{\textfloatsep}{4pt}

\begin{table}[H]
\begin{tabular}{lll}
Variable                            & Mean   & Range          \\ \hline
\% White                                  & 0.89  & (0, 1)         \\ 
\% Female                                    & 0.55  & (0, 1)         \\
\% Urban                                    & 0.42  & (0, 1)         \\
\% High school graduate                  & 0.34  & (0, 1)         \\
\% Poor                                & 0.13  & (0, 1)   \\
\% Moved in last 5 years                     & 0.43  & (0, 1)         \\
\% Households occupied                      & 0.87  & (0.015, 1)      \\
\% Smoke                                & 0.25  & (0.096, 0.44)   \\
Mean Medicare age                           & 74.87 & (68.00, 96.26)    \\
Mean temperature (K)                        & 287.14 & (272.51, 301.14) \\
Mean relative humidity (\%)                & 0.0081 & (0.0033, 0.017) \\
log(Population)                                & 8.24  & (1.39, 11.65)  \\
Population per Square Mile                   & 1307.88  & (0.023, 158503.38)  \\
\hline
Mortality Rate (per PY)                    & 0.046 & (0, 0.39) \\
\end{tabular}
\caption{\textmd{Summary of outcome level covariates from Medicare Beneficiary data ($n=35,036$).}}\label{tab:zip_covs}
\end{table}


\clearpage

\section{Additional Analyses}\label{sec:SensDiffSPercs}
We show the histogram of \dual across our 35,036 ZIP codes in \cref{fig:S_Distro_Dual}. This plot shows percent poor is concentrated towards lower values, and the upper end of \dual shows highly impoverished ZIP codes. 

\begin{figure}[ht!]
\includegraphics[width=1\textwidth]{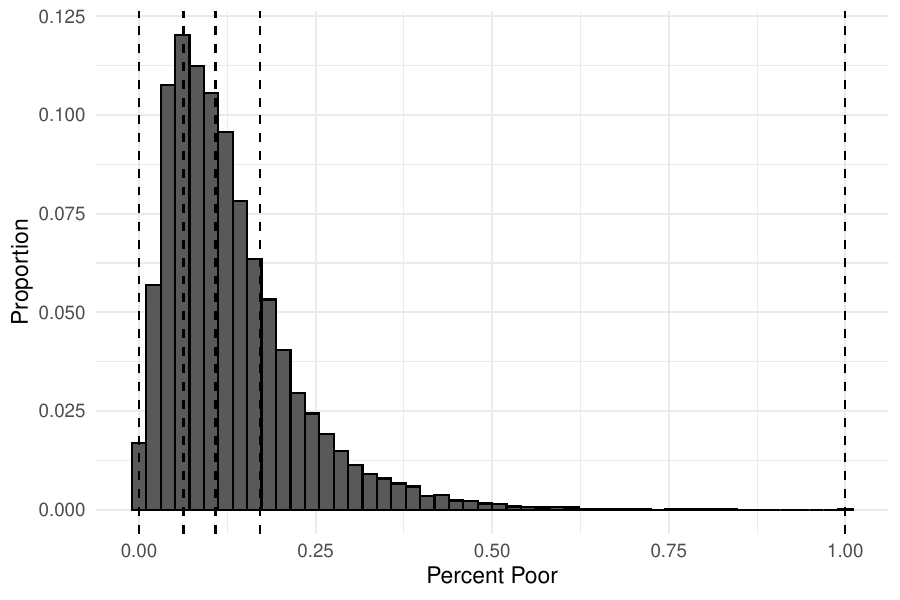} 
\caption{Distribution of \dual across all ZIP codes in the U.S.\label{fig:S_Distro_Dual}}
\end{figure}

We also demonstrate the Total effects on a map, overall and stratified by subgroup \cref{fig:Hist_TEs_PctPoor}, followed by the plot of these estimates on a graph \cref{fig:mapAndHistTEdual}.

\begin{figure}[ht!]
\centering
    \includegraphics[width=1\textwidth]{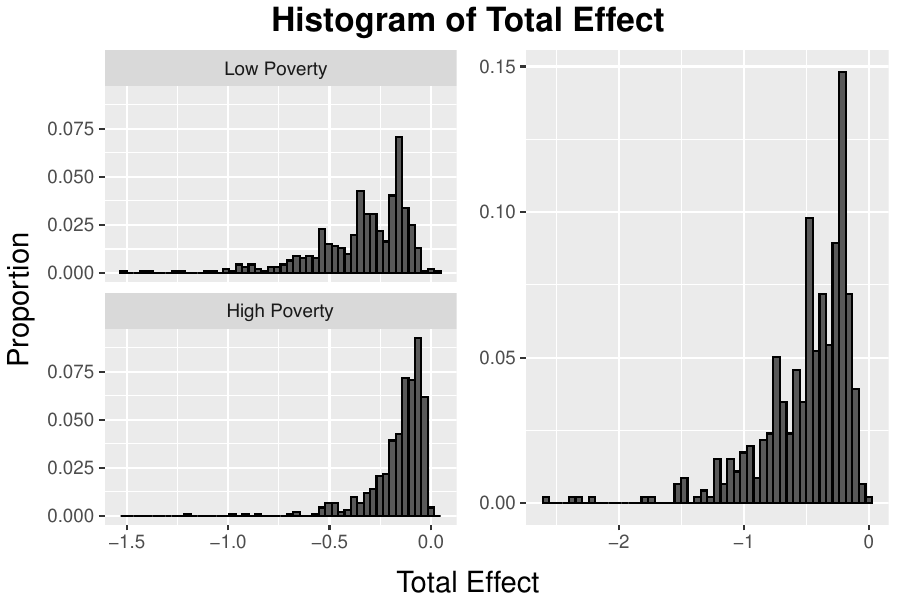} 
\caption{$\EstTotEff_j$ for \dual, the \dualHigh, and the \dualLow. The left shows the overall $\EstTotEff_j$ for \dual, and the right panel shows the breakdown by subgroup with $\EstTotEff(1)$ or \dualHigh on the bottom left and  $\EstTotEff(0)$ or \dualLow on the top left. \label{fig:Hist_TEs_PctPoor}}
\end{figure}

\begin{figure}[ht!]
\centering
\includegraphics[width=1\textwidth]{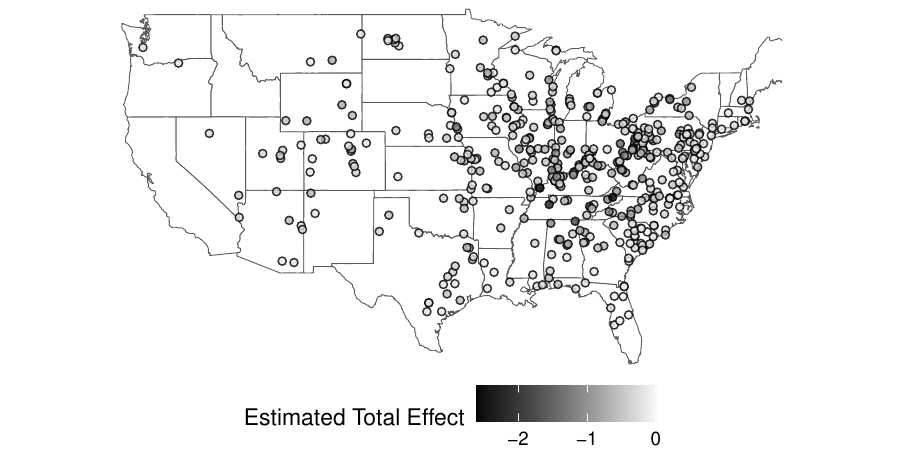} 
\caption{The distribution of $\EstTotEff_j$ on the U.S. map. The darker the bubble is, the greater the reduction on mortality rates the scrubber installation yields. \label{fig:mapAndHistTEdual}}
\end{figure}

\subsection{Comparison of Fair and Optimal Policy Plots}


\begin{figure}[ht!]
\centering
 \resizebox{\textwidth}{!}{
\begin{tabular}{ccc}
    \includegraphics[width=0.33\textwidth]{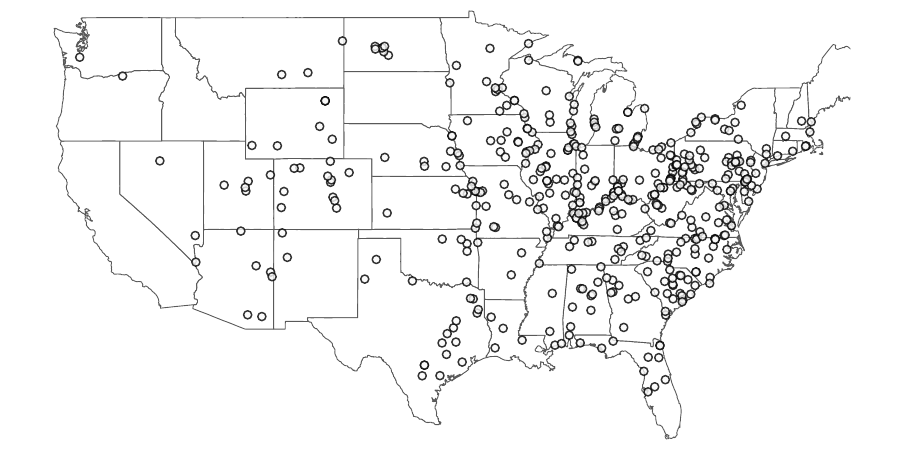}
     &
     \includegraphics[width=0.33\textwidth]{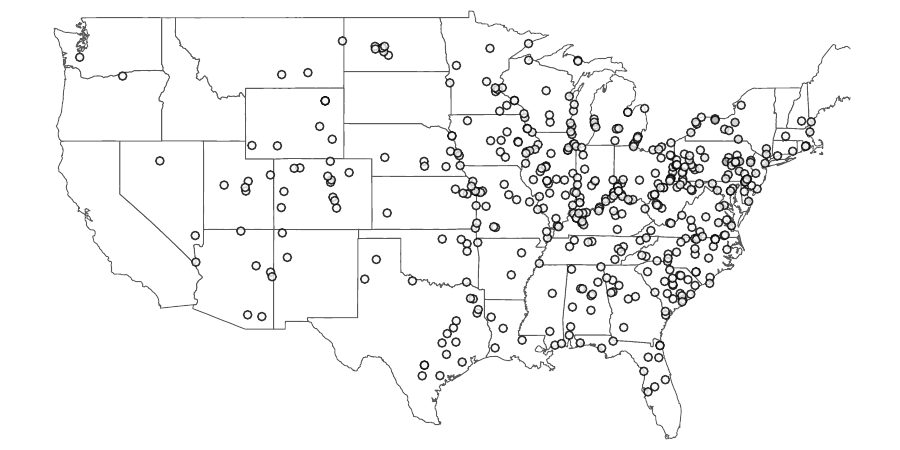} &

     \includegraphics[width=0.33\textwidth]{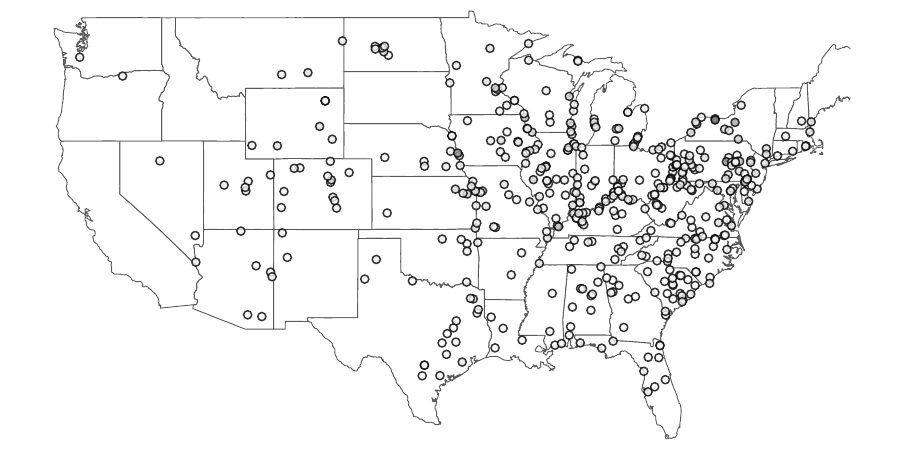}
     
     \\
     (a) 10\% Total Cost & (b) 20\% Total Cost  & (c) 30\% Total Cost 
\end{tabular}}

 \resizebox{\textwidth}{!}{
\begin{tabular}{ccc}
    \includegraphics[width=0.33\textwidth]{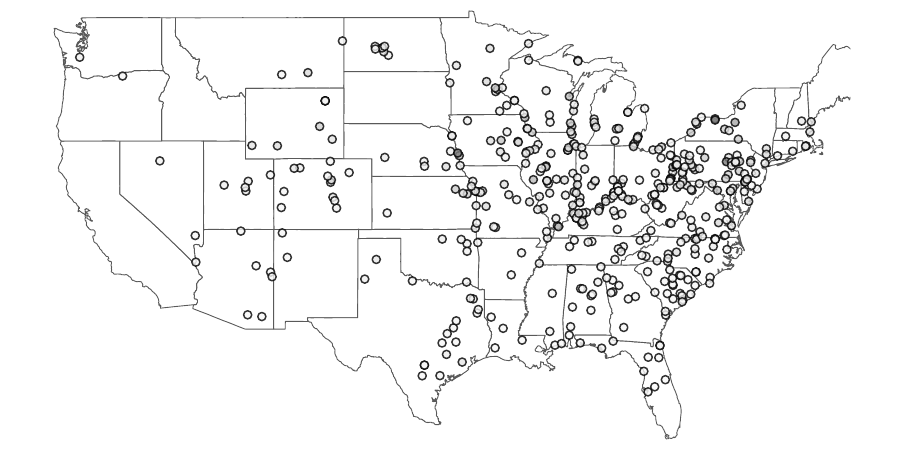}
     &
     \includegraphics[width=0.33\textwidth]{figures/rwd/0.75/policy_cmp/Map_FairPolicy_PctPoor_prop=0.5.pdf} &

     \includegraphics[width=0.33\textwidth]{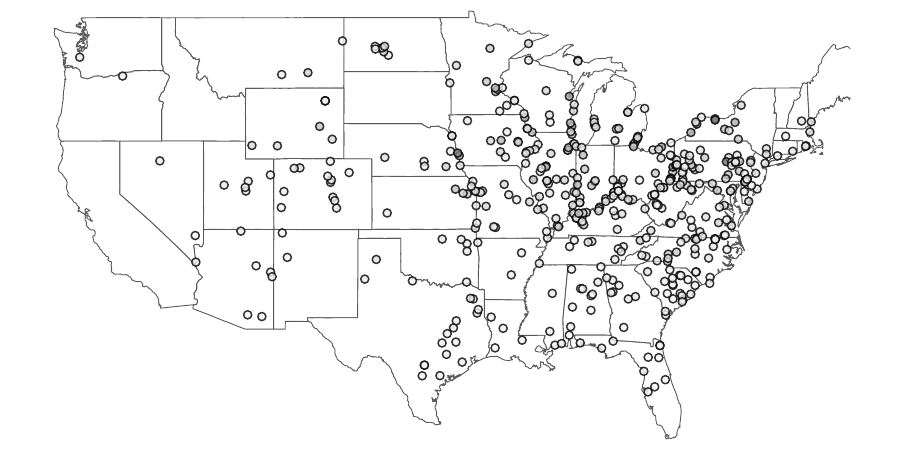}
     
     \\
     (a) 40\% Total Cost & (b) 50\% Total Cost  & (c) 60\% Total Cost 
\end{tabular}}

 \resizebox{\textwidth}{!}{
\begin{tabular}{ccc}
    \includegraphics[width=0.33\textwidth]{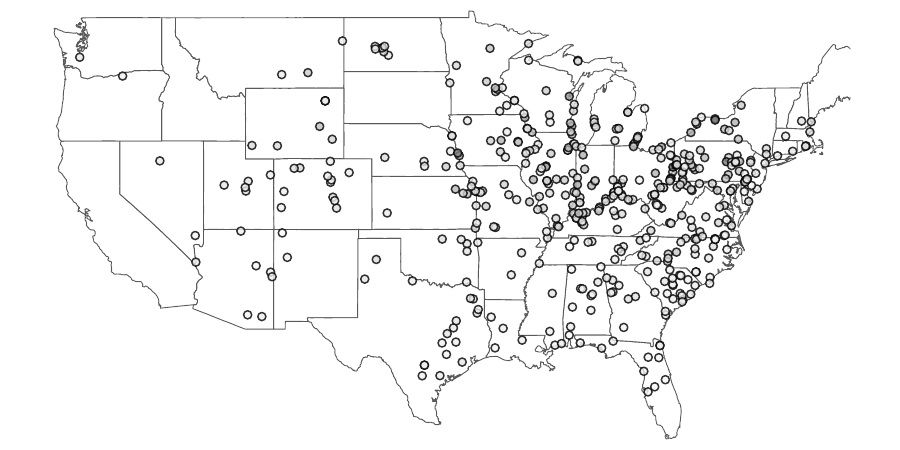}
     &
     \includegraphics[width=0.33\textwidth]{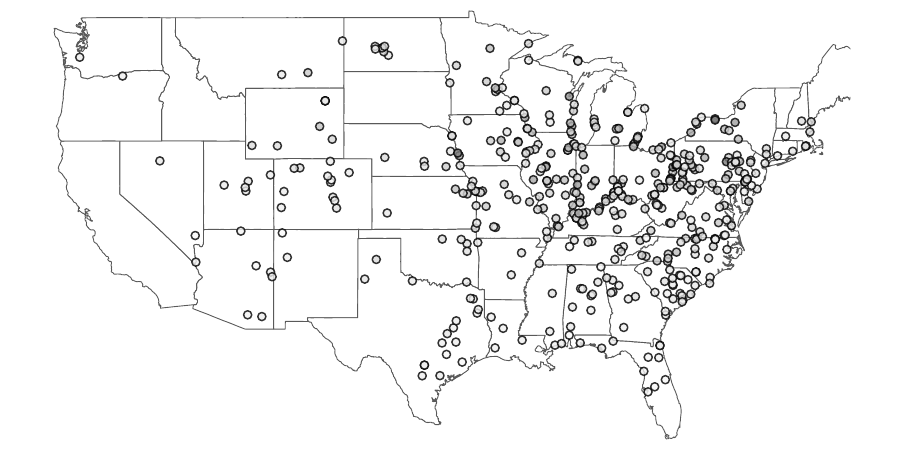} &

     \includegraphics[width=0.33\textwidth]{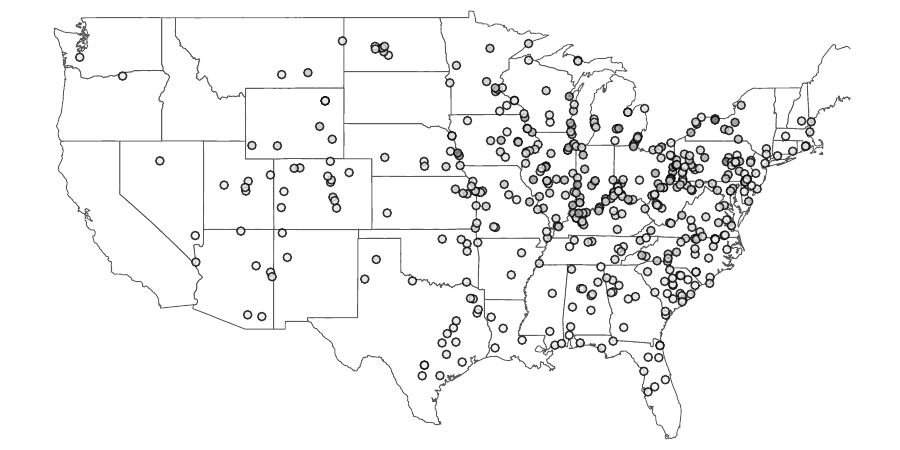}
     
     \\
     (a) 70\% Total Cost & (b) 80\% Total Cost  & (c) 90\% Total Cost 
\end{tabular}}
\caption{\textmd{Grid with the Mortality Reduction from the Fair Policy, varying the spending from 10\%-90\% of budget.}
\label{fig:panelFair}}
\end{figure}

\begin{figure}[ht!]
\centering
 \resizebox{\textwidth}{!}{
\begin{tabular}{ccc}
    \includegraphics[width=0.33\textwidth]{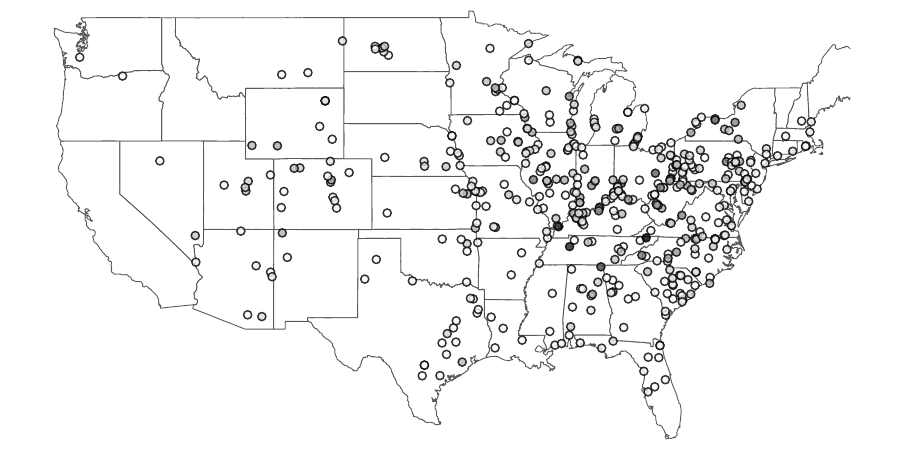}
     &
     \includegraphics[width=0.33\textwidth]{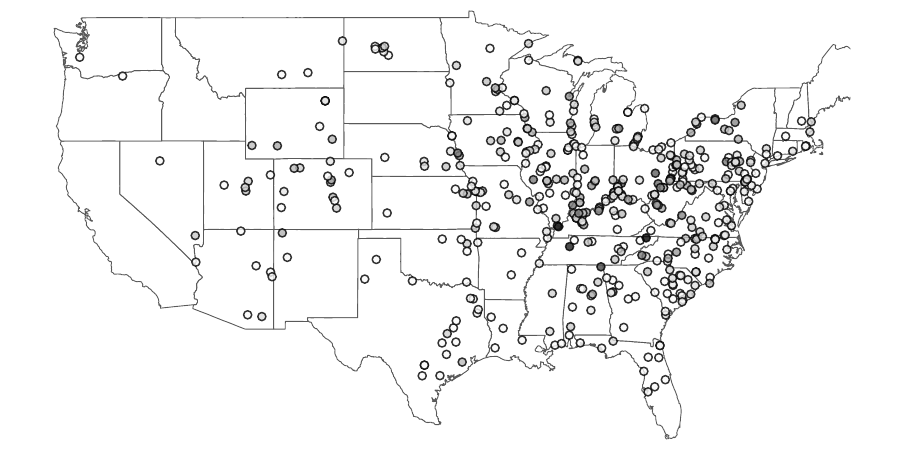} &

     \includegraphics[width=0.33\textwidth]{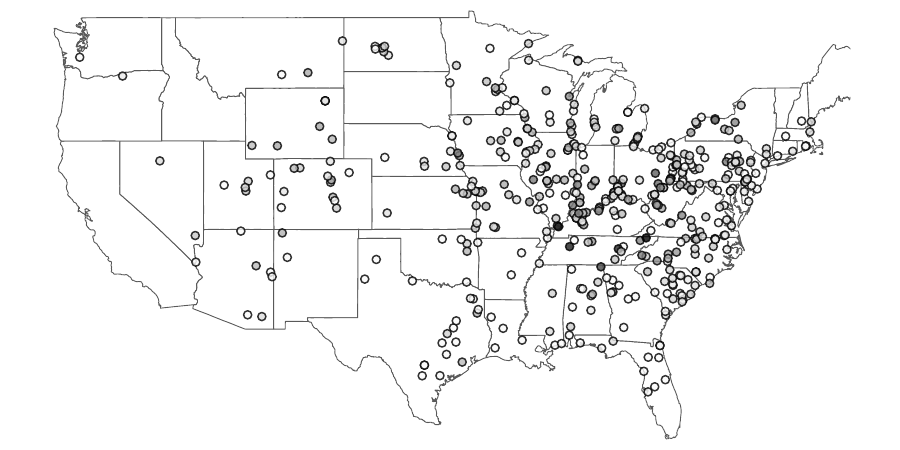}
     
     \\
     (a) 10\% Total Cost & (b) 20\% Total Cost  & (c) 30\% Total Cost 
\end{tabular}}

 \resizebox{\textwidth}{!}{
\begin{tabular}{ccc}
    \includegraphics[width=0.33\textwidth]{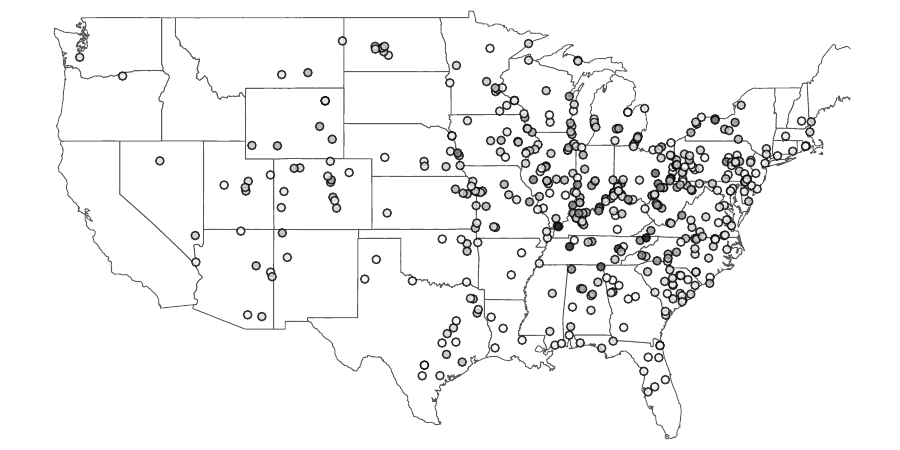}
     &
     \includegraphics[width=0.33\textwidth]{figures/rwd/0.75/policy_cmp/Map_OptimalPolicy_PctPoor_prop=0.5.pdf} &

     \includegraphics[width=0.33\textwidth]{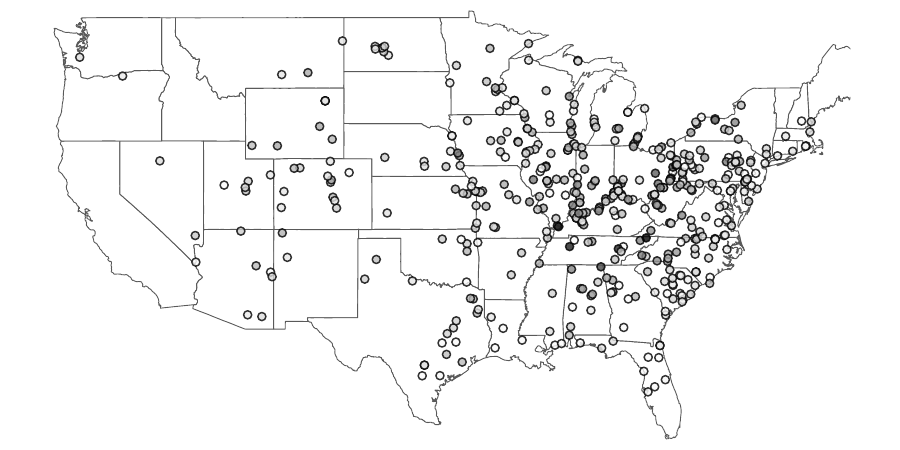}
     
     \\
     (a) 40\% Total Cost & (b) 50\% Total Cost  & (c) 60\% Total Cost 
\end{tabular}}

 \resizebox{\textwidth}{!}{
\begin{tabular}{ccc}
    \includegraphics[width=0.33\textwidth]{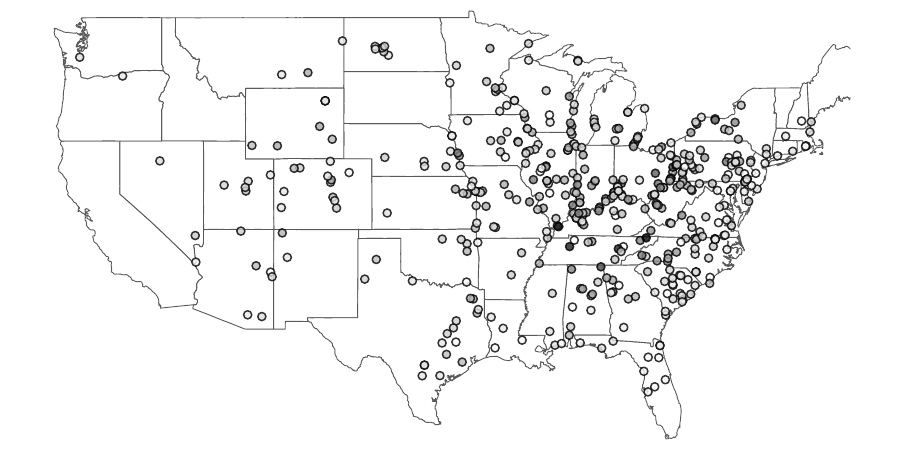}
     &
     \includegraphics[width=0.33\textwidth]{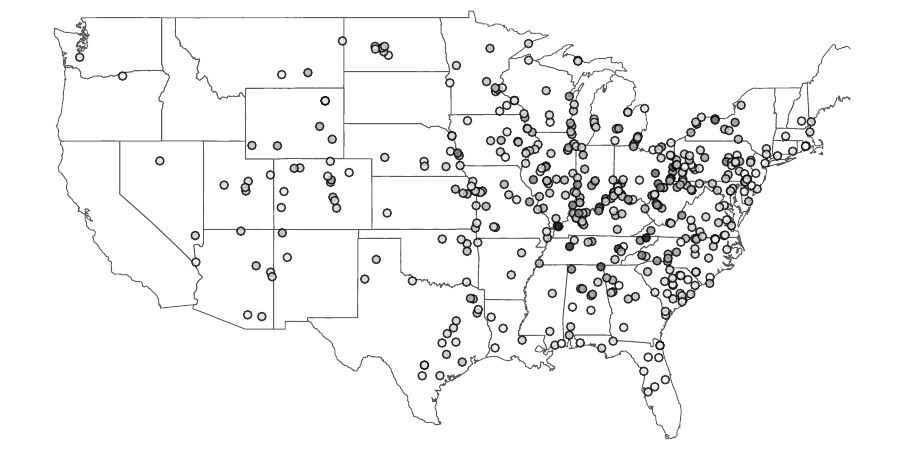} &

     \includegraphics[width=0.33\textwidth]{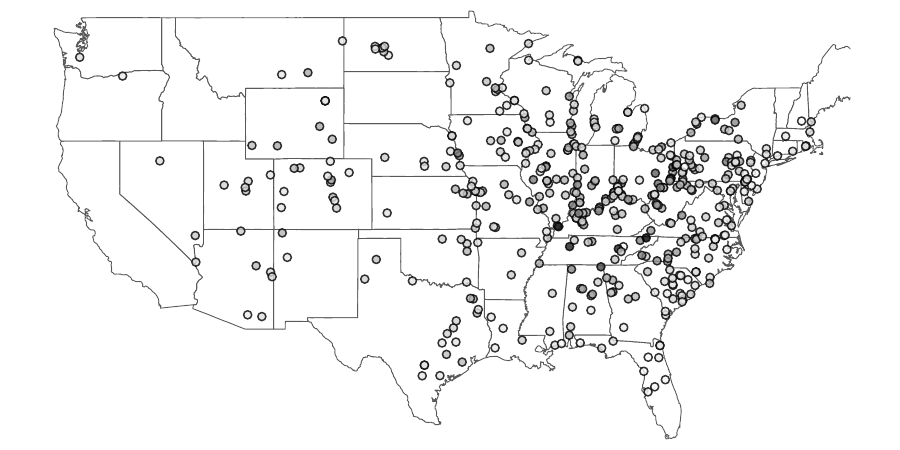}
     
     \\
     (a) 70\% Total Cost & (b) 80\% Total Cost  & (c) 90\% Total Cost 
\end{tabular}}
\caption{\textmd{Grid with the Mortality Reduction from the Optimal Policy, varying the spending from 10\%-90\% of budget.}
\label{fig:panelOptimal}}
\end{figure}






\end{document}